%% file: main.tex
\documentclass[12 pt]{article}
\usepackage{fullpage}
\usepackage{booktabs} 
\usepackage[ruled]{algorithm2e} 

\usepackage[style=alphabetic]{biblatex}
\AtEveryBibitem{
    \clearfield{note}
    \clearfield{issn}
    \clearfield{isbn}
    \clearfield{doi}
    \clearfield{url}
}

\usepackage[utf8]{inputenc}
\usepackage{nameref}
\usepackage{amsmath,amsthm,amssymb}
\usepackage{enumerate}
\usepackage{color}
\usepackage{hyperref}
\usepackage{comment}
\usepackage{caption}
\usepackage{tikz}
\usepackage{qtree}
\usepackage{float}
\usepackage{multicol}
\usepackage{placeins}
\usepackage{booktabs}
\usepackage{enumitem}
\usepackage{graphicx}
\usepackage[T1]{fontenc}
\usepackage[makeroom]{cancel}
\usepackage{tgtermes}
\usepackage{titlesec}
\usepackage{tcolorbox}
\usepackage{soul}
\usepackage{mdframed}
\usepackage{mathtools}
\usepackage{wrapfig}

\SetAlFnt{\small}
\SetAlCapFnt{\small}
\SetAlCapNameFnt{\small}
\SetAlCapHSkip{0pt}
\IncMargin{-\parindent}

\input{header}

\begin{document}
\title{Selecting a Match: Exploration vs Decision}
\author{Ishan Agarwal \thanks{New York University. This work was supported in part by NSF Grant CCF-$1909538$} \and Richard Cole \footnotemark[1] \and Yixin Tao \thanks{London School of Economics. Part of the work done while Yixin Tao was a Ph.D. student at Courant Institute, NYU. This work was also supported in part by ERC Starting Grant ScaleOpt-$757481$.}}
\date{\today}
\maketitle

\input{abstract}

\input{introduction}
\input{related-work}

\input{model}

\input{results}
\input{preliminaries}

\input{lower-bound}

\input{upper-bound}

\input{size-lower-bound}
\input{population-upper-bound}

\input{type-one-upper-bound}

\input{type-two-upper-bound}

\input{imbalance-bound}

\input{ind-hyp-holds-init}

\input{numerics}

\input{open-problems}
\printbibliography
\newpage
\input{appendix}
\end{document}

%% file: header.tex
\definecolor{darkgreen}{rgb}{0.0, 0.5, 0.0}

\newcommand{\rjc}[1]{{\color{black}#1}}
\newcommand{\RJC}[1]{}
\newcommand{\YXT}[1]{}

\newcommand{\IA}[1]{}

\newcommand{\pr}[1]{\operatorname{Pr}\left[{#1}\right]}
\newcommand{\chance}{\operatorname{Pr}}
\newcommand{\Imb}{\operatorname{Imb}}
\newcommand{\Expect}{E}
\newcommand{\dd}{\texttt{d}}

\newcommand{\hide}[1]{}

\newtheorem{theorem}{Theorem}[section]
\newtheorem{thm}{Theorem}[section]

\newtheorem{lemma}[thm]{Lemma}

\newtheorem{claim}[thm]{Claim}

\DeclarePairedDelimiter\ceil{\lceil}{\rceil}

\newcommand{\cg}{\textnormal{\textsl{g}}}

\newcommand\numberthis{\addtocounter{equation}{1}\tag{\theequation}}

%% file: abstract.tex
\begin{abstract}

In a dynamic matching market, such as a marriage or job market, 
how should agents balance accepting a proposed match with the cost of continuing their search? 
We consider this problem in a discrete setting, in which agents have cardinal values and finite lifetimes,
and proposed matches are random.

We seek to quantify how well the agents can do. 
We provide  upper and lower bounds on the collective losses of the agents, 
with a polynomially small failure probability,
where the notion of loss is with respect to a plausible baseline we define.
These bounds are tight up to constant factors.

We highlight two aspects of this work. First, in our model, agents have a finite time in which to enjoy their matches, namely the minimum of their remaining lifetime and that of their partner; this implies that unmatched agents become less desirable over time, and suggests that their decision rules should change over time. 
Second, we use a discrete rather than a continuum model for the population. 
The discreteness causes variance which induces localized imbalances in the two sides of the market.
One of the main technical challenges we face is to bound these imbalances.

In addition, we present the results of simulations on moderate-sized problems for both the discrete and continuum versions. For these size problems, there are substantial ongoing fluctuations in the discrete setting whereas the continuum version converges reasonably quickly.
\end{abstract}

%% file: introduction.tex
\section{Introduction}
\label{sec::intro}
What strategies make sense when deciding whether to commit to a long-term relationship? We are interested
in pairings between members of two sets of agents, such as an employer offering a job and a worker accepting, a woman (or man) proposing marriage to a person of the opposite sex,\footnote{Single-sex marriages could also be studied, but then there would be just one set of agents. In fact, this does not appear to significantly affect our results, but in this work we have focused on the case of two sets of agents.} a landlord agreeing to rent an apartment to a potential renter.

The key feature of these relationships is that the longer they last, the greater the utility they provide;
for simplicity, we assume this utility is linear in the duration of the match. Nonetheless, as a rule agents do not choose to match as soon as they receive a proposal, for different potential partners may provide different utilities. An employer may be supportive or not, a marriage may be happy or not; the possibilities are myriad. Agents seek to assess the utility of a proposed match and then decide whether to accept or keep searching (such an assessment might be implicit). These judgements can be based on some combination of idiosyncratic factors and commonly shared perspectives. Both sides of a potential match are making this assessment, and a match happens only if both sides accept it.

Assessing potential matches takes time and therefore an agent can consider only a relatively small number of potential matches at any one time. In many circumstances, choices are offered on a take it or lose it basis. Typically, job offers are made with a short decision window.
While marriage or its equivalents have many cultural variations, as a rule offers of marriage when made are accepted or declined; it would be unusual to collect multiple offers and only then decide (in the somewhat unlikely event the parties on the other side would be willing to wait).
Again, for simplicity, we assume agents can consider only one match at a time. 

Furthermore, agents are aware of time slipping by. An unemployed worker cannot afford to stay unemployed indefinitely.
Businesses wish to fill open positions promptly as they need workers to carry out the duties of these open positions.
Many men and women appear to want to pair sooner rather than later (whether the pairing is called marriage or not).
We see two forces at work here: one is the ongoing utility from a match, which starts only when the match is formed.
The second is that at least in some circumstances partners become less desirable as they become older.

We are interested in two questions: 
\begin{center}
    What decision rules make sense and how can their effectiveness be measured?
\end{center}
Each potential decision rules provides a balance between the urge to form a match soon so as to have a longer time in which to enjoy it, and the desire to continue searching in the hopes of finding a better match. 

The equilibrium properties of decision rules have been studied previously in models with a continuum population, a continuum model for short~\cite{Adachi03,BurdettC97,BurdettC99,BurdettW98, smith2006marriage,shimer2000assortative,bloch2000two,eeckhout1999bilateral,lauermann2014stable, mcnamara1990job, damiano2005unravelling}.
In these works, agents are assumed to arrive according to a variety of processes, such as a Poisson process.
In some of these works, they are also assumed to use time discounting of future utility.
Either they have infinite lifetimes in which to seek matches or they depart---die---according to another process.
We discuss this in more detail in the related work section below.
Each agent has an intrinsic appeal, a numeric value, called \emph{charm} in Burdett and Coles~\cite{BurdettC99}.
The utility an agent derives from a match is assumed to be an increasing function of their partner's charm.
Agents receive match proposals at a fixed rate and agents either accept or reject a match immediately; for a match to succeed both participating agents must agree to it.
One natural class of agent strategies are reservation strategies; an agent will accept a proposed match exactly if the partner has charm at least $c$. Typically the chosen $c$ is a function of the agent's own charm.
The right choices of reservations $c$ yield equilibrium strategies.

In contrast, we study this problem in a discrete, albeit stochastic, setting.
By this we mean that a finite number of agents arrive at each time step; we also choose time to be discrete.
In addition, we model lifetimes differently, viewing all lives as having duration $T$.
This has the effect of making agents less demanding over time which we believe is a real
effect, and an effect that will not arise with a departure rate that stays the same over time.

Discreteness introduces variance, which leads to localized imbalances in the numbers of men and women
(by localized, we mean agents of a given age and charm). The analysis and bounding of these imbalances
are the largest challenge we face, and while asymptotically small, for moderate values
of our parameters these are non-trivial quantities, as confirmed by our simulation results.
This is in sharp contrast to a continuum setting, where there will be no variance.
Finally, it is not clear that our setting will converge to an equilibrium or near-equilibrium, and while our simulations for moderate parameter values suggest a certain level of stability, they also show that there is continuing substantial variability.
In any event, our concern is to understand the quality of the outcomes:
in a sense we make precise shortly, our model achieves near-optimal utility with high probability.

\paragraph{Roadmap} In section~\ref{sec::model} we formally define our setting, and in section~\ref{sec::results} we state our results. 
Following some preliminaries in section~\ref{sec::prelim},
we present our lower bound in section~\ref{sec::lower-bound}, and outline
the construction for our upper bound in section~\ref{sec::upper_bound}.
In section~\ref{sec::simulations} we describe our simulation results and
we conclude in section~\ref{sec:open-problems} with some additional remarks.
Many proofs are deferred to the appendix.

%% file: related-work.tex
\subsection{Related Work}
\label{sec::related}

Rogerson~\cite{RogersonSW05} surveyed issues of search cost and bargaining in job markets.
More recently, Chade, Eeckout and Smith~\cite{ChadeES17} gave a broad survey of matching in economic models,
covering search with and without costs, and settings with and without transferable utility.
We focus on settings with search costs and no transferable utility.
Even in this domain there are many works.
We characterize these works w.rt.\ multiple dimensions.

The first is the treatment of time, both as regards arrivals and departures.
Most papers assume agents remain in the market till they are matched.
A few allow matches to be broken via a Poisson process (e.g., jobs end, partners divorce) and then the agents
return to the market; see Shimer and Smith~\cite{shimer2000assortative} and Smith~\cite{smith2006marriage}.
Others have agents ending their participation via various random processes:
Burdett and Wright~\cite{BurdettW98} use a Poisson process, Adachi~\cite{Adachi03} uses an exponential random variable, and Lauermann and Noldeke~\cite{lauermann2014stable} use an exogeneous rate.
Arrivals are similarly varied.  Poisson processes in
Burdett and Coles~\cite{BurdettC97}, Smith~\cite{smith2006marriage}, and Shimer and Smith~\cite{shimer2000assortative}. Cloning: when agents leave due to a match they are replaced by clones thereby keeping the available matches unchanged; see Adachi~\cite{Adachi03} and Burdett and Wright~\cite{BurdettW98}.
Fixed arrival rates: see Eeckhout~\cite{eeckhout1999bilateral}, and Lauermann and Noldeke~\cite{lauermann2014stable}. Finally, no new arrivals: see Damiano, Hao and Suen~\cite{damiano2005unravelling}, and McNamara and Collins~\cite{mcnamara1990job}.

The second dimension is the choice of utility model.
These are all functions of the partner's charm, though there is considerable variation.
The most common is that the utility an agent gains is a non-decreasing function, either linear~\cite{BurdettC97} or more general~\cite{smith2006marriage,eeckhout1999bilateral}; some papers allow for time discounting~\cite{Adachi03,bloch2000two};
the utility can be the product of the partners' charms~\cite{damiano2005unravelling};
or it is given by independent random variables for each pair of agents~\cite{BurdettW98,mcnamara1990job};
another option is that the agents obtain their utility by dividing
a reward which is a function of their individual charms~\cite{shimer2000assortative}.

The final dimension is the choice of equilibrium model.
Most of the papers consider a steady state equilibrium;
McNamara and Collins~\cite{mcnamara1990job} consider Nash Equilibria,
and Damiano, Hao and Suen~\cite{damiano2005unravelling} analyze
a multi-round dynamic equilibrium.

The tension between taking a choice now and waiting for potentially better options arises in multiple other domains,
including secretary problems\cite{Ferguson89}, online matching\cite{KarpVV90}, matching market thickening~\cite{AkbarpourLG20,BaccaraLY20},
and regret minimization~\cite{blumlearning}.
In spirit, the secretary problem seems the most analogous as it involves a single decision, albeit by just a single agent. We discuss it briefly in the next paragraph. In contrast, online matching has a centralized decision maker that seeks to optimize the outcome of many choices. Regret minimization occurs in a distributed setting, however here each agent makes multiple decisions over time, with the goal of achieving  a cumulatively good outcome; again, this seems quite distinct from our setting.
Market thickening is used in contexts where a global matching is being computed, which seem unlike the random matches on offer in our setting.

The standard secretary problem is expressed in terms of ranks.
A cardinal version was considered by Bearden~\cite{Bearden06}; 
here the goal is to maximize the expected value of the chosen
secretary, with values uniform on $[0,1]$.
For each applicant the decision maker learns whether they are the best so far.
Bearden shows the optimal strategy is to reject the first $\sqrt{n}-1$ candidates, 
and then choose the first candidate
who meets the ``best so far'' criteria.
Clearly, the expected value of the selected secretary is
$1-\Theta(1/\sqrt{n})$,
which is analogous to the bounds we obtain, although the settings appear quite distinct.
Bearden argued that the payoff rule in this version of the
problem is more natural that the classic version.
The problem of maximizing the duration of a relatively best choice
has also been considered~\cite{Ferguson89}.

%% file: model.tex
\section{The Model}
\label{sec::model}

We consider a setting in which, at each time step, $n$ agents enter a matching pool. Agents exit the pool either when they are matched or if they have been in the pool for $T$ time steps. There are two types of agents, called men and women.  Each match pairs a man with a women. At each time step the agents are paired uniformly at random. Each pair comprises a proposed match. Each agent in a pair can accept or reject the proposed match as they prefer; a match occurs only if both agents accept it. 

In a discrete setting, a random pairing seems more natural than having pairs arrive one by one, for the process
of pairing will proceed in parallel, and pairs are necessarily mutually exclusive.
While in practice the pairings under consideration at any one time will not cover the
whole of the smaller side of the population, considering a maximal matching seems a reasonable simplification.

We assume agents evaluate their potential partners using
cardinal values, and furthermore these are common values: every agent of the opposite type (gender) has the same value $v_i$ for agent $i$. In the terminology of Burdett and Coles, this is agent $i$'s charm.

We associate two parameters $v_i$ and $t_i$ with agent $i$. $v_i$ is the agent's charm and $t_i$ is the total time remaining before agent $i$ is forced to exit the pool. Agent $i$ derives utility $v_j\cdot \min(t_i,t_j)$ when matched with agent $j$.
We assume that the values lie in the range $[T,2T)$, and that an agent's value, chosen when it enters the pool, is one of $\{T, T+1,\ldots, 2T-1\}$, picked uniformly at random. 
We note that the relative utilities of an agent are scale free;
in other words, the range assumption is equivalent to assuming the values lie in the range $[1,2]$.
We could have used a separate discretization for the values, but we preferred to avoid an additional parameter. Furthermore, it would not affect the results qualitatively.

Entering agents are either male or female with equal probability. 

Throughout this work it will be useful to view the market as a $T\times T$ size box, with agents
located at grid points. The box is indexed by value
on the horizontal axis and
by time on the vertical axis. 
Consider the set of $T$ points on the top edge: $\{(T,0),(T+1,0),\ldots,(2T-1,0)\}$. Agents enter the market at one of these points, picked uniformly at random.  At each time step, an agent either matches and leaves the box or moves down vertically by $1$ unit. After $T$ steps, if unmatched for all these times, the agent exits the box (at the bottom).

\paragraph{A Reasonable Notion of Loss}
In a single gender version of this setting, the total utility derived by the $n$ agents that enter at any one time step is at most $n \cdot \sum_i v_i\cdot T$; in the two-gender case, by applying a Chernoff bound, one can obtain a similar bound with high probability.
This bound can easily be achieved if all agents simply accept whatever match is proposed to them in the very first step in which they enter the matching pool. However such behavior seems implausible for high value agents, as their expected utility would be much smaller than what they might reasonably hope to achieve. Consequently, we set $v_i \cdot T$ as a reasonable target for $i$'s achieved utility. Based on this, we define the \emph{total loss} suffered by the agents to be:
$$\sum_{\substack{i\in \{\text{agents obtaining utility}\\~~~~~~~\text{less than their worth}\}}} v_i \cdot T-\text{utility obtained by agent } i.$$

This measure captures the intuition that agents who obtain less than their worth due either to a lower value partner, or to accepting a match only later on in the process, are suffering losses. We want to capture how much utility is lost compared to the benchmark in which each agent gets an equal value partner for the whole length $T$ time period. It also addresses what is implausible about the naive solution, in which all agents immediately accept whatever match is proposed to them, and which maximizes the usual notion of social welfare. 

It is not clear how to determine an optimal strategy, let alone whether it can be computed feasibly.
For a truly optimal strategy would incorporate the effects of past variance, a level of knowledge that seems implausible in practice;
and even an ex-ante optimal strategy seems out of reach.
Instead, we will present a strategy, which we call the \emph{reasonable strategy}, which seeks to ensure that if it is followed by all the players, then the total loss will be at most a constant factor larger than what could be achieved by the optimal strategy. Actually, we introduce two strategies, and the second one, called the \emph{modified reasonable strategy}, is the one we analyze.

%% file: results.tex
\section{Results}
\label{sec::results}

We obtain a lower bound on the total loss suffered by agents; no matter their behavior, they will, with high probability, suffer an average loss of $\Omega(T\sqrt{T})$.

\begin{theorem}
\label{thm::lb-loss-two_sex}
Suppose the matching market runs for $\tau$ time steps. If $16\leq T+1\leq n, c\geq1,T\leq \tau \leq n^c$, and $n\geq 96T(2c+2)\ln n$, then, over $\tau$ time steps, whatever strategies the agents use, with probability at least $1-\frac{1}{4n^c}$, 
the average loss per agent is at least $\frac{T\sqrt{T}}{20}$.
\end{theorem}

On the other hand, we construct a strategy profile, which if followed by all the agents, leads, with high probability, to a total loss of at most $\mathrm{O}(T\sqrt{T})$.

\begin{theorem}
\label{thm::ub-loss-two_sex}
 Suppose $2T\leq\tau\leq n^c$, $c\geq 1$, $676\leq T$, and $n \geq (3654 + 2436e^{12} + 546(e^{12} + 1) c)^2(3c+4) T^3 (\log_2 n)^2 \ln n.$ Then, over $\tau$ time steps, if all agents follow the modified reasonable strategy, with probability at least $1-\frac{1}{n^c}$, 
the average loss per agent is at most $11T\sqrt{T}$.
\end{theorem}

Our results hold for large $n$ and $T$. Furthermore, Theorem \ref{thm::ub-loss-two_sex} applies only when $n$ is much larger than $T$. However, our numerical simulations suggest that similar results hold even for quite moderate values of $n$ and $T$ and also do not require $n$ to be much bigger than $T$.
To simplify the presentation, we assume that $T= 4^i$ for some integer $i>0$, though the bounds extend to all values of $T$, possibly with somewhat larger constants.

%% file: preliminaries.tex
\section{Preliminaries}
\label{sec::prelim}
We review the notion of negative cylinder dependence and make a simple observation regarding the matching procedure.

\begin{lemma}
\label{lem::match rate}
Suppose there are $m$ men and $w$ women in total.
Further suppose that for a given man $x$, there are $w'$ women for which a proposed
match would be accepted by both sides.
Then a random match will provide man $x$ such a match with probability 
$w'/\max\{m,w\}$.
\end{lemma}
\begin{proof}
If there are at least as many women as men, every man will be offered a match,
and the probability that it is accepted by both sides is $w'/w$.
While if there are more men, a man will be offered a match with probability
$w/m$, and thus the probability that he is offered an acceptable match
is $w/m \cdot w'/w = w'/m$.
\end{proof}

\paragraph{Negative Dependence}
Consider a set of $0$-$1$ valued valued random variables $\{X_i\}_{i=1}^{n}$. The set $\{X_i\}$ is $\lambda$-\emph{correlated} if
\begin{align*}
E\Big[ \prod_{i=1}^{n} X_i\Big]\leq\lambda\cdot\prod_{i=1}^n E\Big[X_i\Big],
\end{align*}
where $\lambda\geq 1$.  
The set $\{X_i\}$ is \emph{negative cylinder dependent} if $\{X_i\}$ and $\{1-X_i\}$ are both $1-$correlated. In our arguments we will apply Chernoff-like bounds to negative cylinder dependent variables. We will use the following lemmas; their proofs are deferred to Appendix \ref{appn::prelim}.

\begin{lemma}
\label{lem::negative_dependence_two_sex}
Let $S_m$ and $S_w$ be two sets of $N_1$ and $N_2$ agents respectively, Suppose that $N_1\leq N_2$. Let $S_a=\{a_1,a_2,\ldots, a_n\}\subseteq S_m$ and $S_b =\{b_1,b_2,\ldots, b_r\} \subseteq S_w$. Consider a matching between $S_m$ and $S_w$ chosen uniformly at random.
Let $X_i$ be an indicator variable which equals $1$ if agent $a_i$ is paired with an agent in $S_b$, and $0$ otherwise. Then the set $\{X_i\}$ is negative cylinder dependent and for any $\delta>0$,
\begin{align*}
    \chance\left[\sum X_i \geq (1 + \delta) \mu \right] \leq e^{- \frac{\delta^2 \mu}{3}} \text{~~and~~}
    \chance\left[\sum X_i \leq (1 - \delta) \mu \right] \leq e^{- \frac{\delta^2 \mu}{2}}.
\end{align*}
\end{lemma}

\begin{lemma} \label{lem::chernoff_bound}
If $\{X_i\}_{i = 1}^n$ are $1$-correlated random variables taking value $\{0, 1\}$ and $\overline{\mu}$ is an upper bound on $\mu = E[\sum X_i]$, then, for any $\delta>0$,
\begin{align*}
    \chance\left[\sum X_i \geq (1 + \delta) \overline{\mu} \right] \leq e^{- \frac{\delta^2 \overline{\mu}}{3}} \text{~~and~~}
    \chance\left[\sum X_i \leq \mu-\delta \overline{\mu} \right] \leq e^{- \frac{\delta^2 \overline{\mu}}{2}}.
\end{align*}
\end{lemma}

%% file: lower-bound.tex
\section{Lower Bound on the Loss for Any Strategy}
\label{sec::lower-bound}

The intuition for this result is fairly simple. 
If an agent remains unmatched for $\sqrt{T}$ steps, then any subsequent proposed match would cause a loss of at least $\sqrt{T}$ to one of the participating agents.
Thus to avoid having average losses of
$\Omega(\sqrt{T})$,  most matches would need to occur during an agent's first $\sqrt{T}$ steps.

But we will show that for at least a constant fraction of the agents, the matches they are offered during their first $\sqrt{T}$ steps will all have the property that the values of the two agents differ by at least $\sqrt{T}$, and consequently one of the participating agents would suffer a $\sqrt{T}$ loss.
The overall result follows.

This second claim is not immediate because the probability that an agent is offered a close-in-value match might vary significantly from agent to agent and over time.

We somewhat optimize constants and consequently consider a period of time $w = \Theta(\sqrt{T})$ and value differences $w$, instead of precisely the value $\sqrt{T}$ used in the above outline.

\begin{proof} (Of Theorem~\ref{thm::lb-loss-two_sex})~
We divide the grid into width $w$
columns, where a column includes the low-value side boundary, but not the high-value boundary; one end column may be narrower. We will set the parameter $w$ later.

We consider the set of proposed matches at some arbitrary time $t$.
We say a proposed match is \emph{safe} if the paired agents are in the same or adjacent columns.
We also define the male match rate $p_i$ for column $i$ to be the probability that a man in the column has a safe match. By Lemma~\ref{lem::match rate}, this is at most the number of women in columns $i-1$, $i$, and $i+1$ divided by the maximum of the total number of women and the total number of men, which is at most the number of women in these columns divided by the total number of women. 
Clearly the sum of the male match rates over all the columns is at most $3$.
The same claim holds for the analogous female match rates.

Consider the men entering the system at time $t$, which we call the \emph{new} men.
Each column contains at most $w$ points at which agents enter the market, namely the points along the column's top edge, and each entering agent is equally likely to be a man or a woman. By applying a Chernoff bound, we see that for any given column $i$,

\vspace*{-0.1in}

$$\chance\Big[\text{\# of new men in column $i$ at time $t$} \ge \frac{(1+\delta)nw}{2T} \Big]\le e^{-\delta^2nw/6T}.$$

Applying this bound to every column over $\tau$ consecutive time steps yields:
\begin{align*}
&\chance\Big[\text{every column receives at most } \frac{(1+\delta)nw}{2T} \text{ new men} \\[-10pt]
&\hspace*{1.5in}\text{for each of $\tau$ consecutive time steps}\Big] 
\ge\Big(1-e^{-\delta^{2}nw/6T}\Big)^{\tau T/w}.
\end{align*}
Call this event $\mathcal E$.
Henceforth we condition on $\mathcal E$.

Now suppose that every time an agent was offered a safe match, they accepted it. Recall that $p_i$ is the match rate for column $i$.
By Lemma \ref{lem::negative_dependence_two_sex}, for the new men at time $t$ in column $i$, for any $t$,
$$\chance\Big[\text{\# safely matched men } \le \frac{(1+\delta)n w p_i}{2T}\Big] \ge1-e^{-\delta^2 n w p_i/6T}.$$

In fact, agents may not accept every proposed safe match; but this only reduces the number of agents safely matched, and therefore the bound on the probability continues to hold.

Furthermore, by Lemma~\ref{lem::chernoff_bound}, letting $\overline{\mu} = \frac{n w \max\{p_i, \frac{w}{T}\}}{2T}$, gives 
\begin{align*}
    \chance\Big[\text{number of safely matched men } \le \frac{(1+\delta)n w \max\{p_i, \frac{w}{T}\}}{2T}\Big] \ge1-e^{-\frac{\delta^2 n w \max\{p_i, \frac{w}{T}\}}{6T}} \ge 1-e^{-\frac{\delta^2 n w^2}{6T^2}}.
\end{align*}


Recalling that $\sum_i p_i \le 3$, and applying a union bound over all $T/w$ strips
for $w$ successive steps,
we obtain, for any given set of new men
entering at some time $t$,
over their first $w$ time steps,
\begin{align}
\label{eqn::prog-num-safe-match}
\chance\Big[\text{\# of safely matched men} \leq\frac{2(1+\delta)nw^2}{T}\Big]\ge1- \frac{T}{w} w e^{-\delta^2nw^2/6T^2}.
\end{align}

In addition, for any given set of new men, on applying a Chernoff bound, we know that
\begin{align}
\label{eqn::prob:num-new-men}
\chance\Big[\text{\# of new men} \geq\frac{n(1-\epsilon)}{2}\Big]\ge 1 - e^{-\epsilon^2n/4}. 
\end{align}
For each remaining man in each of the first $\tau-w$ sets of new men---of which there are at least $(\tau - w)\big(\frac{n(1-\epsilon)}{2} - \frac{2(1 + \delta)nw^2}{T}\big)$---
one of the following two cases must apply.
\begin{itemize}
    \item He has not been matched after spending $w$ time in the system. Now, if and when he is matched, the only way he can avoid suffering a $wT$ loss is to match with a sufficiently higher value woman. In this case the higher value woman suffers at least a $wT$ loss.
    \item He has been matched within $w$ time but it was not a safe match. In such a match whichever agent had the higher value suffered at least a $wT$ loss.
\end{itemize}

Since the system runs for $\tau$ time steps, this argument can be applied to all agents except those that enter the system during the last $w$ time steps. 
We deduce that the total loss generated by all these agents is at least
$(\tau-w)(\frac{n(1-\epsilon)}{2}- \frac{2(1 + \delta)nw^2}{T})\cdot wT$.

Note that this loss is being shared by up to $n\tau$ agents.
Hence there is an average loss of at least $\frac 12\big(wT(1-\epsilon) - \frac{4(1 + \delta)w^3T}{T}\big)\cdot\frac{\tau-w}{\tau}.$
Setting
$w = \frac{\sqrt{T}}{4}$, 
and using the lower bound on $\tau$ ($\tau\geq T$), we obtain:
$$\text{average loss per agent}\geq
\frac 12 \Big(\frac{T\sqrt{T}(1-\epsilon)}{4}-\frac{T\sqrt{T}(1+\delta)}{16}\Big)\cdot\Big(1-\frac{\sqrt{T}}{4T}\Big).$$

Now we set $\delta=\sqrt{\frac{6T^2}{nw^2}\ln(3n^{c}T\tau)}$ and $\epsilon=\sqrt{\frac 4n\ln(3\tau n^{c})}$. We would like to have $\delta\leq1$, which we enforce by our choice of constraints on $n,T,\tau$ and $c$ (namely $16\leq T\leq n, c\geq1,T\leq \tau \leq n^c$ and $n\geq 96T(2c+2)\ln n$). These constraints also  ensure that $\epsilon\leq 1/16$. Substituting $\delta \le 1$ and $\epsilon\leq 1/16$ yields: 
$$\text{average loss per agent}\geq\frac{7T\sqrt{T}}{128}\cdot\frac{15}{16}\geq \frac{T\sqrt{T}}{20}. $$

By \eqref{eqn::prog-num-safe-match}
and \eqref{eqn::prob:num-new-men},
this bound holds with probability at least
\begin{align*}
\chance[{\mathcal E}] \cdot\Big(1-\tau T  e^{-\frac{\delta^2 nw^2}{6T^2}}-\tau e^{-\frac{\epsilon^2 n}{4}}\Big)
& \ge  \Big( 1  - \frac{1}{n^c}\Big).
\end{align*}
The detailed calculation can be found in Appendix \ref{appn::loss_prob_calculation}.
\end{proof}

%% file: upper-bound.tex
\section{Upper Bound on the Loss when Using the Modified Reasonable Strategy}
\label{sec::upper_bound}

The lower bound suggests that plausible agent strategies will yield a constant probability of matching
every $\sqrt{T}$ steps. This would imply that the number of agents present decreases geometrically
with agent age; more precisely, there would be a constant factor decrease for every $\sqrt{T}$ increment in
age. Then, in order to maintain match probabilities, all agents would have to be willing to match
with young agents who will accept them. In fact, the decreases we just described are far from uniform,
which makes the analysis quite non-trivial. Nonetheless, the above intuition informed the design of the following agent strategies.
The first strategy, which we call ``a reasonably good strategy'' seems quite natural, but for ease of analysis we consider a modified strategy which we prove to be asymptotically within a constant factor of optimal.

We define the \emph{worth} of an agent to be $v_i\cdot (T-t_i)$;
this is the maximum utility its partner
could derive from a match with this agent.
Note that the worth of an agent decreases as it ages.

\vspace*{-0.05in}

\paragraph{A Reasonably Good Strategy}
In this strategy an agent accepts a proposed match if it gives the agent utility at least $v_i\cdot (T-t_i)\cdot(1-\frac{1}{\sqrt{T}}-\frac{t_i}{T})$. 
The terms $1/\sqrt{T}$ and $t_i/T$ are present to approximately balance
the expected loss of utility from not matching in a single step with the marginal gain in utility agent $i$ could receive from being more demanding in terms of the minimum worth it will accept in a partner.

\vspace*{-0.05in}

\paragraph{The Modified Reasonable Strategy}
We partition the $T\times T$ size space into the regions defined below,
as shown in Figure~\ref{fig:strips}.
In the modified strategy, an agent accepts a proposed match exactly if the proposed partner lies
in the same region.
This partition uses regions of two kinds, which we call \emph{strips}. 

\begin{itemize}
\item 
Type $1$ strips: these are strips that have new people entering the strip at the top. The $i$-th Type $1$ strip is defined as the region between the parallel lines $v=2(t-1)+T+(i-1)\sqrt{T}$ and $v=2(t-1)+T+i\sqrt{T}$; they have $\sqrt{T}$ width and $\sqrt{T}/2$ height. Points on the first (left) line are included in the strip, but points on the second (right) line are excluded. There are $\sqrt{T}$ Type $1$ strips.

\item 
\emph{Type $2$ strips}: these strips do not touch the top boundary of the box. The strips are again defined by parallel lines. They have successive heights $\sqrt{T}$,
$\sqrt{T}$, $2\sqrt{T}$, and then repeatedly doubling up to $T/2$.
Here the points on the first (upper) line are excluded from the strip and the points on the second (lower) line are included in the strip. There are $\log_2 \sqrt{T} +1$ Type $2$ strips.
\end{itemize}

\begin{figure}[bht]
\centering
	\includegraphics[scale=0.45]{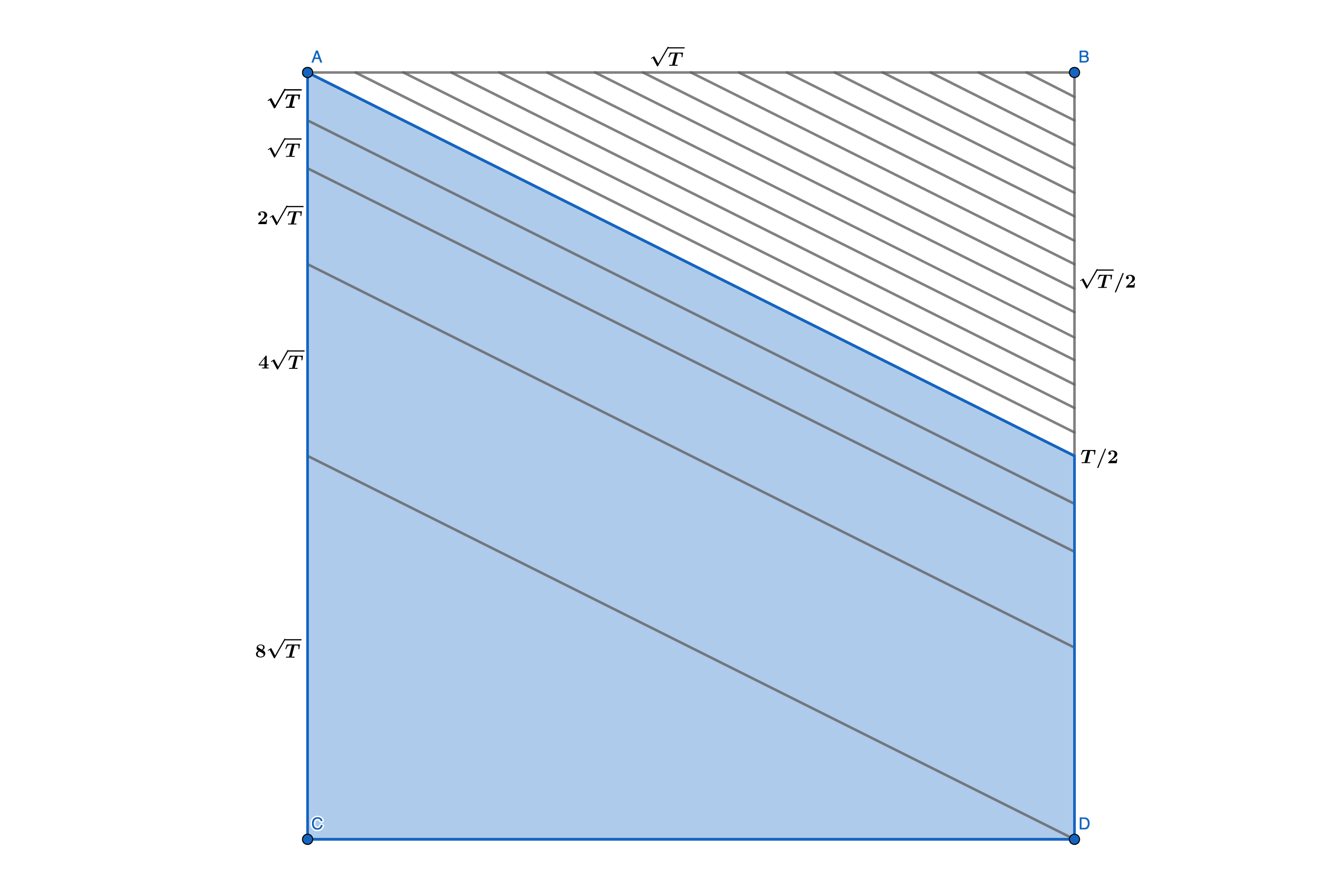}
	\caption{The two types of strips used to partition the matching pool.}
	\label{fig:strips}
\end{figure}

We note that with the previously stated reasonable strategy, agents would be willing to match with some agents outside their strip and would reject some agents in the same strip. However, using the modified strategy simplifies the analysis, for if all agents use the modified strategy, agents will definitely get accepted when they accept a match. We will prove that the modified strategy is not much worse than the optimal strategy in terms of the average loss of value suffered by an agent. 

\vspace*{-0.05in}

\paragraph{Outline of the proof of the upper bound}
Our analysis assumes the following constraints on $n$ and $T$.
\begin{align}
\label{eqn::constraints}    
\begin{array}{l}
c\geq 1,\\
T \ge 676,\\
n \geq (3654 + 2436e^{12} + 546(e^{12} + 1) c)^2(3c+4) T^3 (\log_2 n)^2 \ln n. 
\end{array}
\end{align}

The result follows from a high-probability inductive bound on the overall population, the strip populations, and the male-female imbalances in each strip. 
We start at time $t=0$. Time $t$ will refer to the moment after the new agents have entered in this step, but before the match occurs.

\begin{lemma}\label{lem::ind-bound}
Let $N$ denote the total number of strips.
Suppose that the constraints in \eqref{eqn::constraints} hold.
Then, with probability at least $1-1/n^c$, the following
inductive hypothesis $H(t)$ holds at the start of 
every time step $t$, immediately following the entry of the new agents at time $t$, for $\sqrt{T}\le t \le n^c$.

\begin{enumerate}
\item\label{itm::tot-pop}
The total population is at most $\frac{3}{2}nN+n$.
\item The population of every Type $1$ strip is at most $2.6n$.
\item The population of every Type $2$ strip is at most $\frac{7.5n\sqrt{T}}{\text{maximum height of the strip}}$.
\item The population in the bottommost Type $2$ strip is no more than $60n/\sqrt{T}$.
\item\label{itm::pop-imb}
In every strip $s$, except possibly the bottommost Type 2 strip,
the imbalance, $\Imb(s,t) = \big|\text{the number of men in $s$} 
-\text{the number of women in $s$}%
\big| 
\leq n/25\sqrt{T}$.
\end{enumerate}
\end{lemma}
\begin{proof} (Sketch.)~
We will show in Theorems~\ref{thm::total_size_upper_bound}
and \ref{type1_upper}--\ref{thm::imbalance_bound}
that each of the above five clauses 
holds with high probability. 
The last of these results also requires a high-probability lower bound, Theorem~\ref{thm::lb-size},
on the population size in the same time range.
In addition, in Theorem~\ref{thm::initialization}, we show that, with high probability, the inductive hypothesis is true initially. Summing the failure probabilities prove the lemma. This calculation can be found in Appendix \ref{appn::imbalance}.
\end{proof}

With this result in hand we can upper bound the average agent loss.

\subsection{The Theorems and Proof Sketches}
Let $\widetilde{\mathcal{E}}$ be the event that the inductive hypothesis $H(t)$ holds at the start of 
time step $t$ immediately following the arrival of the new agents in this step, for $\sqrt{T}\le t \le n^c$.

\subsubsection{Bounding the loss}

We first bound an individual agent's loss based on its match time. We then obtain an overall bound on the loss.  As argued below, Theorem~\ref{thm::ub-loss-two_sex} follows immediately.

\begin{lemma}
\label{lem::loss_in_strip}
In the modified reasonable strategy, if an agent with value $v$ matches at time $t$, its utility loss is at most
$4Tt +2t\sqrt{T}$.
\end{lemma}
This result follows by a simple calculation based on the strip geometry. The proof is in Appendix \ref{appn::loss_from_match}.

\begin{theorem}
\label{thm::upper_bound_on_loss}
Suppose the constraints in \eqref{eqn::constraints} hold.
Also, suppose that all agents follow the modified reasonable strategy.
In addition, suppose the system runs for $\tau \ge 2T$ time steps, where $\tau\leq n^c$. 
Then the average loss per departing agent over these $\tau$ steps will be at most $11 T\sqrt{T}$.
\end{theorem}

\begin{proof} 
Consider the first $\tau$ time steps of the matching process.
Let $n_i$ denote the number of agents who match and thereby leave the pool at age $i$ during these $\tau$ steps. By Lemma~\ref{lem::loss_in_strip},
each such agent suffers a loss of at most
$4Ti+2T\sqrt{T}$. Thus the total loss is bounded by:
\begin{align*}
\text{Total loss} \le \sum_{i=0}^{T-1} \left(4 Ti\cdot n_i + 2T\sqrt{T}\cdot n_i\right).
\end{align*}

Each agent who is matched at age $i$ is present in the matching
pool for $i+1$ steps. By clause~\ref{itm::tot-pop} of the inductive hypothesis in Lemma~\ref{lem::ind-bound}, 
at each time during this period, the population of the matching pool is at most $\frac{3}{2}nN+n\leq \frac{3}{2}n(\sqrt{T}+\log_2\sqrt{T}+1)+n\leq 2n\sqrt{T}$, where the last inequality follows from $\sqrt{T} \le 26$ due to constraint \eqref{eqn::constraints}.
Thus, 
\begin{align*}
\sum_{i=0}^{\tau-1} (i+1) n_i \le 2n\sqrt{T}\cdot \tau.
\end{align*}
Therefore,
\begin{align*}
\text{Total loss} \le 8nT\sqrt{T}\cdot \tau + \sum_{i=0}^{T-1}  2T(\sqrt{T}-2)\cdot n_i.
\end{align*}

Let $D\triangleq \sum_{i=0}^{\tau-1} n_i$, the number of agents that leave during the first $\tau $ steps.
We observe that $D$ is at most $n\tau$, the number of agents that entered during this period.
Also, as the population of the pool at any time is at most $ 2 
n\sqrt{T}$, we see that $D\geq n\tau- 2 
n\sqrt{T}$. By assumption, $\tau\ge 2T$ and $\sqrt{T}\ge 26$, 
so
\begin{align*}
\frac {12}{13}n\tau \le D \le n\tau.
\end{align*}

This yields the following bound on the total loss:
\begin{align*}
    \text{Total loss}\leq 8 
    n\tau T\sqrt{T}+ 2n\tau T\sqrt{T} \le 10 
    n\tau T\sqrt{T}.
\end{align*}

And therefore,
\begin{align*}
    \text{Average loss per agent}= \frac{\text{Total loss}} {D}
    \leq \frac{10 
    nT\tau\sqrt{T}} {\frac{12}{13} 
    n\tau}< 11 T\sqrt{T}.
\end{align*}
\end{proof}

\begin{proof}(Of Theorem~\ref{thm::ub-loss-two_sex})~
This follows immediately from Lemma \ref{lem::ind-bound} and Theorem \ref{thm::upper_bound_on_loss}.
\end{proof}

%% file: size-lower-bound.tex
\subsubsection{Total Size Lower Bound}

\begin{theorem}
\label{thm::lb-size}
Suppose $H(t)$ and the constraints in~\eqref{eqn::constraints} hold.
If all agents follow the modified reasonable strategy, then with probability at least $1- 1/n^{2c+1}$,
for every time $t\in[\sqrt{T}, n^c]$, the population in the matching pool is at least $\frac 13 n\sqrt{T}$.
\end{theorem}
\begin{proof} (Idea.)~
We consider only the new agents that entered the matching pool over the last $\sqrt{T}$ time steps. We then bound how many of these agents could have been matched in this time period. Suppose that at any particular time step $t$, the match rate experienced by the men in strip $i$ is $p_i$. The critical observation is that the sum of the $p_i$ is at most $1$. The same is true for the women. This allows us to prove that even if we could set the match rates in an adversarial manner, only about $n/\sqrt{T}$ of the agents that entered at any one time could be matched in any single time step (in the discussion here, we neglect the effects of variance). This allows us to show that, of the agents we consider, only about $\sum_{i=1}^{\sqrt{T}}in/\sqrt{T} \approx n\sqrt{T}/2$ could have been matched over the last $\sqrt{T}$ time steps. This provides a lower bound on the total size of roughly $n\sqrt{T}/2$. Accounting for the variance that can occur when achieving a high probability bound causes the bound on the number of matches to degrade to $n\sqrt{T}/3$. 
The full proof can be found in Appendix \ref{appn::lower_bound_on size}.
\end{proof}

%% file: population-upper-bound.tex
\subsubsection{Population Upper Bound}
\label{sec::total_size_upper_bound}

\begin{theorem}
\label{thm::total_size_upper_bound}
Suppose $H(t)$ and the constraints in \eqref{eqn::constraints} 
 hold. If all agents follow the modified reasonable strategy, then at the start of time step $t+1$,
with probability at least $1 - 1/n^{2c+1}$,
the total population of the matching pool will be at most $(3/2)nN+n$, where $N$ is the total number of strips.
\end{theorem}
\begin{proof} (Idea.)~
We seek to lower bound the number of matches in one time step. 
If it exceeds the number of incoming agents, then the total population reduces.
The expected number of matches is minimized when the strip populations are equal, and on applying Lemma~\ref{lem::match rate}, this yields the following
lower bound on the number of matched women (or men):
$[N\cdot (P/2N)^2/(P/2)]= P/(2N)$, where $P$ is the upper bound on the population.
This yields the condition $P/N \le n$, or $P\le nN$.
The argument is completed by taking account of the deviations needed to ensure a high-probability bound. The full proof can be found in Appendix \ref{appn::total_upper}.
\end{proof}

%% file: type-one-upper-bound.tex
\subsubsection{Upper Bound on the Size of a Strip.}

We begin with a technical lemma. 
\begin{lemma} \label{lem::upper_strip_tech}
Let $s$ be a strip, and let $S$ be an arbitrary subset of the men and women in $s$. Let $m$ be the number of men and $w$ be the number of women in $S$. In addition, let $X$ be the imbalance for the whole of $s$. Then the expected number of people in $S$ that are matched in a single step is at least 
\begin{align*}
    \frac{\frac{(m + w)^2}{2} - \frac{X^2}{2}}{\max \{\text{\# of men, \# of women}\} \text{ in the whole population } }
\end{align*}
\end{lemma}
\begin{proof}
We need only consider the case that $|X| \leq m + w$. \footnote{Otherwise, the bound is negative.} Let $m_{t}$ denote the total number of men in this strip and $w_t$  the total number of women. In addition, let $\Delta \triangleq m - w$, $P \triangleq m + w$, $Q \triangleq m_t + w_t$ and $X = m_t - w_t$. Then, $m = \frac{P + \Delta}{2}$, $w = \frac{P - \Delta}{2}$, $m_t = \frac{Q + X}{2}$ and $w_t = \frac{Q - X}{2}$. The expected number of people matched in this subset of men and women is 
\begin{align*}
    \frac{m w_t + m_t w}{\max \{\text{\# of men, \# of women}\} \text{ in the whole population } }.
\end{align*}
We now focus on the numerator: $m w_t + m_t w = (PQ - X \Delta) / 2 = (P^2 + P(Q - P)  - X^2 - X (\Delta - X)) / 2$. In order to show this is larger than $\frac{P^2}{2} - \frac{X^2}{2}$, it suffices to show $P(Q - P) \geq X (\Delta - X)$. 

As $m_t \geq m$ and $w_t \geq w$, $Q - P \geq \Delta - X$ and $Q - P \geq X - \Delta$.
Recall that it suffices to consider the case $|X| \leq m + w = P$. Combining these two inequalities yields $P(Q - P) \geq X (\Delta - X)$, which proves the result.
\end{proof}

Next, we give an upper bound on the size of a Type $1$ strip.
\begin{theorem}
\label{type1_upper}
Suppose $H(t)$ and the constraints in \eqref{eqn::constraints} hold. If all agents follow the modified reasonable strategy, then at time $t+1$, right after the new agents have entered,
with probability $>1-\frac{1}{n^{2c+1}}$, each Type $1$ strip will continue to have population at most $dn$,  where $d=2.6$.
\end{theorem}

\begin{proof} (Sketch).
Consider a strip $s'$ and its successor strip $s$ (the strip immediately to its left). We will follow the collection of agents occupying $\sqrt T$ adjacent diagonals over $\sqrt{T}$ steps, beginning with the at most $dn$ agents in strip $s'$ and ending in strip $s$, with the remainder of these agents plus any new agents who have entered these diagonals. The heart of our proof is to show that in a single step
we maintain the $dn$ bound on the number of agents in this collection of advancing diagonals. The basic idea is straightforward:
we compute a lower bound on the expected number of matches using Lemma~\ref{lem::upper_strip_tech}
taking into account the maximum possibly imbalance, add
the incoming agents and correct for variance.
One more important detail is that the expected number of matches
is minimized if, in the collection of agents we are tracking,
half are in strip $s$ and half are in $s'$;
so this is the value we use in these calculations.
The actual proof can be found in Appendix \ref{appn::type1_upper}.
\end{proof}

%% file: type-two-upper-bound.tex
\begin{theorem}
\label{type2_upper}
Suppose $H(t)$ and the constraints in  \eqref{eqn::constraints}  hold. If all agents follow the modified reasonable strategy, 
then at time $t+1$, right after the new agents have entered,
with probability $>1-\frac{1}{n^{2c+1}}$,
each Type $2$ strip (apart from the bottommost one) will continue to have population at most $\frac{\cg n\sqrt{T}}{\text{height of the strip}}$ where $\cg=7.5$.
\end{theorem}
\begin{proof} (Idea.)~
For the topmost Type 2 strip $s$ we obtain a bound of $2\cdot2.6n =5.2 n \sqrt{T}/\sqrt{T}$,
as the items in $s$ are obtained from its predecessor strip over
the previous $\sqrt{T}$ steps, i.e.\ the sum of the contents
at times $\sqrt{T}/2$ and $\sqrt{T}$ earlier.
The same bound applies to the second Type 2 strip.
Each subsequent Type 2 strip $s$ has twice the height of its
predecessor. Let $s$ have height $H$.
The contents of $s$ come from its predecessor over a period of length $H$, which by the inductive hypothesis contain at most $2\cg n\sqrt{T}/(H/2) = 4\cg n\sqrt{T}/ H$ agents.
To prove our bound, we need to show at least $3\cg n\sqrt{T}/H$
of them are removed during these $H$ steps.
Again, as in Theorem~\ref{type1_upper},
we seek to track a population as it moves from $s'$ to $s$.
The challenge is that in the analysis this population shrinks
over time and the match rate is proportional to the
square of this population. 
To get a fairly tight bound, we formulate this as a differential
expression and determine the smallest value for the
constant $\cg$ that enables this number of matches.
The full proof can be found in Appendix \ref{appn::type2_upper}.
\end{proof}

\begin{theorem}
\label{thm::last_strip_size_upper_bound}
Suppose $H(t)$ and the constraints in  \eqref{eqn::constraints} hold.
If all agents follow the modified reasonable strategy, 
then at time $t+1$, right after the new agents have entered,
the strip population for the bottommost Type $2$ strip will continue to be at most $\frac{60n}{\sqrt{T}}$.
\end{theorem}
This simple calculation is deferred to Appendix \ref{appn::type2_upper}.

%% file: imbalance-bound.tex
\subsubsection{Bound on Imbalance}

\label{sec::imbalance_bound}
\begin{theorem}
\label{thm::imbalance_bound}
Suppose that $H(\tau)$ and the constraints in \eqref{eqn::constraints} hold. If all agents follow the modified reasonable strategy, 
then with probability at least $1-2/n^{2c+1}$, in every strip $s$ (except possibly the bottommost Type $2$ strip), 
$\Imb(s) \le n/25\sqrt{T}$.
\end{theorem}
\begin{proof}
We divide each strip into thin diagonals of width $1$. Let the diagonal include the bottom but not the top boundary. Notice that for each value, a diagonal contains at most one grid point. 

We introduce the following notation w.r.t.\ diagonal $d$ at time step $\tau$, where we are conditioning on the outcome of step $\tau-1$.
\begin{align*}
I(d,\tau) &= \Expect[(\text{number of men at time $\tau$}-\text{number of women at time $\tau$})]\\
X(d,\tau) &=  (\text{number of men matching at time $\tau$}-\text{number of women matching at time $\tau$})\\
&\hspace*{0.2in} - \Expect[(\text{number of men matching at time $\tau$}-\text{number of women matching at time $\tau$})]\\
Y(d,\tau) &= \text{number of men entering at time $\tau$} - \text{number of women entering at time $\tau$}
\\
A(d,\tau) &= (\text{number of men matching at time $\tau$}+\text{number of women matching at time $\tau$})/2.
\end{align*}

$I(d,\tau)$ is measured after the entry of the new agents at time $\tau$ but prior to the match for this step. Also, note that $Y(d, \tau) = 0$ if $d$ is in a Type $2$ strip. 

In addition, observe that the imbalance $\Imb(s)$ at the start of step $t$ equals $\sum_{d\in s} I(d,t)$.


We observe that a match between two agents in distinct diagonals of the same strip
will increment the $(\text{number of men } - \text{ number of women})$
in one diagonal and decrement it in the other.
Thus there is a zero net change over all the diagonals
in the strip due to the matches. However, as the agents all age by 1 unit during a step, some agents enter the strip and some leave, which can cause changes to the imbalance within a strip.
However, the entry of new agents can introduce new imbalances.
We will need to understand more precisely how these imbalances evolve.

It is convenient to number the diagonals as $d_1,d_2,d_3,\ldots$, in right to left order.

\begin{claim}
\label{clm::update-to-I}
Let $d_i$ and $d_j$ be two diagonals in the same strip $s$. For brevity, let $I_i\triangleq I(d_i,\tau-1)$,
$I_j\triangleq I(d_j,\tau-1)$,
$A_i\triangleq A(d_i,\tau-1)$,
$A_j\triangleq A(d_j,\tau-1)$,
$X_i\triangleq X(d_i,\tau-1)$,
$X_j\triangleq X(d_j,\tau-1)$.
Finally, let $R$ denote the maximum of the total number
of men and the total number of women in the system
at time $\tau-1$.

Then the new imbalance on diagonal $d_i$, prior to every unmatched agent adding 1 to their age (which causes the agents on $d_i$ to move to $d_{i+1}$), denoted by $I'(d_i,\tau)$, is given by:
\begin{align*}
&I'(d_i,\tau)= \\
&~~~~I_i + X_i - \sum_{d_j \in s}\Big[   X_i\frac{(2A_j-I_j-X_j)}{4R}-X_j\frac{(2A_i-I_i-X_i)}{4R}+I_i\frac{(2A_j-I_j-X_j)}{4R}-I_j\frac{(2A_i-I_i-X_i)}{4R}\Big]; \\
&\text{~~and~~}I(d_i, \tau) = I'(d_{i-1}, \tau - 1) + Y(d,\tau).
\end{align*}
\end{claim}

This claim is shown by considering the expected number of matches involving agents in diagonals $d_i$ and $d_j$.
The proof can be found in Appendix \ref{appn::imbalance}. 

The expression $X_i(2A_j-I_j-X_j)/4R$ reflects the reduction
of the contribution of $X_i$ to the total imbalance on diagonal $d_i$ and the corresponding increase on diagonal $d_j$.
Thus it is convenient to view the multiplier $(2A_j-I_j-X_j)/4R$ as indicating the fraction of $X_i$ that is being moved to diagonal $j$; the remaining fraction of $X_i$ remains on $d_i$.

$X(d,\tau)$ and $Y(d,\tau)$ are generated at diagonal $d$ at time $\tau$. In each subsequent time step the portion on each diagonal where it is present will be further redistributed:
\begin{enumerate}
    \item Due to the expected matching at time $\tau'\ge \tau$, each portion of $X(d,\tau)$ and $Y(d,\tau)$ spreads to other diagonals in the same strip.
    \item At the end of  time step $\tau'$ the portions of $X(d,\tau)$ and $Y(d,\tau)$ present on diagonal $d_i$ move to diagonal $d_{i+1}$.
\end{enumerate}

Building on these observations, we will show our bound on the imbalance by means of the following two arguments. Specifically, we show that: 

\begin{enumerate}
    \item For any $\tau$ and $\tau'$, the total contribution from $X(\cdot, \tau)$ and $Y(\cdot, \tau)$ to strip $s$ at time $\tau'$ is bounded.
    \item For times $\tau'\ge \tau + \Omega(T\log n)$, the remaining portions of $X(\cdot, \tau)$ and $Y(\cdot, \tau)$ in the market are small. 
\end{enumerate}

\paragraph{Bound on the contribution of $X$ to the strip $s$}

Notice that  $\sum_{d_i\in s} I'(d_i,\tau) = \sum_{d_i\in s} I(d_i,\tau-1)$, for the coefficients
multiplying $X_i$ cancel, as they also do for $I_i$. Thus we can think of this process as redistributing the imbalance, but not changing the total imbalance.

Over time an imbalance $X(d_i,\tau)$ will be redistributed over many diagonals. We write
$X(d_i,\tau,d_j,\tau')$ to denote the portion of
$X(d_i,\tau)$ on diagonal $d_j$ at time $\tau'$.
$d_j$ need not be in the same strip as $d_i$.
Note that $\sum_{d_j} X(d_i,\tau,d_j,\tau') = X(d_i,\tau)$ for all $\tau'\ge \tau$. $Y(d_i,\tau,d_j,\tau')$ is defined analogously.

An important property concerns the relative
distribution of the $X(d_i,\tau,d_j,\tau')$ and
the $X(d_k,\tau,d_j,\tau')$. In a sense made precise in the following claim, if $k>i$ the $d_k$ terms
remain to the left of the $d_i$ terms.

For the purposes of the following claim, we treat the final strip as a single diagonal, and in addition ignore the fact that people depart at age $T$ (which means that once an imbalance appears in this strip it remains there). The reason this strip is different is that it covers the whole of the bottom boundary and so is the only strip from which people leave the system by aging out.

\begin{claim}
\label{clm::distr-of-X}
For all $\ell$, for all $i<k$, and for all $\tau'\ge \tau$,
$\big|\sum_{j>\ell} X(d_i,\tau,d_j,\tau')\big| \le \big|\sum_{j>\ell} X(d_k,\tau,d_j,\tau')\big|$.
The same property holds for the $Y(d_i,\tau,d_j,\tau')$.
\end{claim}
\begin{proof}
We prove the result for the $X$ terms by induction on $\tau'$; the same argument applies to the $Y$ terms.
Clearly the property holds for $\tau'=\tau$.
Let $x_{ij} \triangleq X(d_i,\tau,d_j,\tau')/X(d_i,\tau)$, and define $x_{kj}$ analogously. Our claim states that $\sum_{j>\ell} x_{ij} \le \sum_{j>\ell} x_{kj}$; we need to show it holds at time $\tau'+1$ also.
We view the $x_{ij}$ as sitting on the unit interval, with $x_{ij}$ taking a portion of length $x_{ij}$, ordered by increasing $j$, and likewise for the $x_{kj}$.
We map aligned portions of the $x_{ij}$ and $x_{kj'}$ to each other.
This mapping has the property that the $j$ index in the $x_{ij}$ term is always equal to or smaller than the $j'$ index in the $x_{kj'}$ term.

Let's look at how aligned portions of $x_{ij}$ and $x_{kj'}$ are dispersed
in the next step. If they are in distinct strips, then $j< j'$ and this property is maintained for all the dispersed portions.

We view the multiplier $(2A_j-I_j-X_j)/4R$ in Claim~\ref{clm::update-to-I} as specifying the fraction of $X_i$ that moves from diagonal $i$ to diagonal $j$. Notice that this multiplier is the same for every diagonal in this strip.

We also note that $I_i$ consists of a sum of terms $X(d_i,\tau,d_j,\tau')$ 
and $X(d_i,\tau,d_j,\tau')$ for diagonals $d_j$ in the same strip as $d_i$ or to the right of $d_i$. Furthermore, the multiplier $(2A_j-I_j-X_j)/4R$ specifies
the fraction of each of these terms that moves from diagonal $i$ to diagonal $j$. Thus if $d_j$ and $d_{j'}$ are in the same strip, the $X$ terms corresponding to the aligned portions of $x_{ij}$ and $x_{kj'}$ are redistributed identically, thereby maintaining the property for these fragments. Naturally, the property also continues to hold for undispersed fragments.

Finally, shifting down by one diagonal, as is done following the dispersal, will leave the property unaffected.

\end{proof}

Later, we will show a common bound $B$ on the sums
$\big| \sum_{i\le j \le k} X(d_j,\tau)\big|$,
which holds for all $d_i$ and $d_k$ in the same strip and all\footnote{The calculation for the bound proved in Claim \ref{clm::bound-on-X-one-strip::B} only applies to | $\sum_{i\le j \le k} X(d_j,\tau)|$, where $\tau>\sqrt{T}$. However for times in the initial $\sqrt{T}$ steps, the bound is only better. A calculation of this bound for times in this initial period is done in the proof of Theorem \ref{thm::initialization}; see Claim \ref{clm:: B_for_init} in Appendix \ref{appn::init}.} $\tau$.
With this bound and Claim~\ref{clm::distr-of-X} in hand, for each strip $s$, we can bound the contribution of the $X(d_i,\tau,d_j,\tau')$ summed
over all $d_i$ and over $d_j\in s$ by $2B$.
\begin{claim}
\label{clm::bound-on-X-one-strip}
For all $\tau'\ge \tau$, for every strip $s$,
$\big|\sum_{d_i; d_j\in s} X(d_i,\tau,d_j,\tau')  \big| \le 2B$.
\end{claim}
\begin{proof}
Let $d_{r(s)}$ be the rightmost (lowest index) diagonal in $s$ and $d_{l(s)}$ be the leftmost (highest index) diagonal in $s$. Let $w_i = \sum_{j\ge r(s)} X(d_i,\tau,d_j,\tau')/X(d_i,\tau)$.
Let's consider 
$\sum_{d_i\in s'; j \ge r(s)} X(d_i,\tau,d_j,\tau')
= \sum_{d_i\in s'} w_i\cdot X(d_i,\tau)
$. 
Notice that $\sum_{r(s')\leq i\leq l(s)} X(d_i,\tau)=0$. By Claim~\ref{clm::distr-of-X}, $w_i \le w_k$, for $i<k$. Thus, 
\begin{align*}
    \Big|\sum_{d_i\in s'; j \ge r(s)} X(d_i,\tau,d_j,\tau')\Big|=
    \Big|\sum_{d_i\in s'} w_i\cdot X(d_i,\tau)\Big| &\leq
    \sum_{r(s') \le r < l(s')} (w_i -w_{i-1})\Big|\sum_{r\le i \le l(s')} X(d_i,\tau)\Big|\\
    &\le (w_{l(s')}-w_{r(s')})\cdot \max_{r \geq r(s')}\Big| \sum_{r\le i \le l(s')} X(d_i,\tau)\Big|.
\end{align*}

We apply this bound to the diagonals from every strip to obtain:
\begin{equation}
\label{eqn::canceled_out_spread}
    \begin{aligned}
    \Big|\sum_{d_i; j \ge r(s)} X(d_i,\tau,d_j,\tau')\Big|=
    \Big|\sum_{s'}\sum_{d_i\in s'; j \ge r(s)} X(d_i,\tau,d_j,\tau')]\Big| \leq \sum_{s'}(w_{l(s')}-w_{r(s')})\cdot B\leq B.\
    \end{aligned}
\end{equation}

Using the same argument, $\big|\sum_{d_i; j \ge l(s)+1} X(d_i,\tau,d_j,\tau')\big| \le B$, since $l(s)+1=r(s'')$ where $s''$ is the strip immediately below $s$. Therefore,
%
%
\begin{align*}
   \Big|\sum_{d_i; d_j\in s} X(d_i,\tau,d_j,\tau') \Big| =
   \Big| \sum_{d_i; j \ge r(s)} X(d_i,\tau,d_j,\tau') -
   \sum_{d_i; j \ge l(s)+1} X(d_i,\tau,d_j,\tau') \Big| \le 2B.
\end{align*}
\end{proof}

\hide{
Notice that the coefficient of the $X_i$ term (as well as the $I_i$ term) in \rjc{Claim~\ref{clm::update-to-I}} is $\frac{(2A_j-I_j-X_j)}{4R}$. This is independent of $i$. So equal fractions of $X_i$ get sent to the diagonal $d_j$ no matter which diagonal $d_i$ we consider in the strip. \RJC{This was already used so can probably be omitted.}
}


\begin{claim}\label{clm::bound-on-X-one-strip::B}
For any time $\tau \leq n^c$, with probability at least $1 - \frac{1}{n^{2c+1}}$, $B \leq 96\Big[ \frac{n\ln(4n^{3c+1} (T^2/32 + T/8) N)}{\sqrt{T}}\Big]^{1/2}$.
\end{claim}
\begin{proof}
First we bound $|\sum_{d\in S} X(d,\tau)|$ for any subset $S$ of consecutive diagonals in a strip $s$. Suppose the total number of men in $S$ is $m$ and the total number of women is $w$.

By Theorem \ref{thm::lb-size}, the total population is at least $1/3 \cdot n\sqrt{T}$. By Theorem \ref{thm::total_size_upper_bound}, it is at most $3nN/2+n$. In addition, by the inductive hypothesis, the total imbalance is bounded by the bottommost strip population plus the individual strip imbalances, and this is at most $60 n / \sqrt{T} + 25 n N / \sqrt{T}$. Therefore,
$$\frac{n\sqrt{T}}{6}\leq\max
\Big\{
\begin{array}{l}
   \text{total number of men},\\ \hspace*{0.2in}\text{total number of women}
\end{array}
\Big\}
\leq \frac 12 \Big(\frac{3n(\sqrt{T}+\log_2\sqrt{T} + 1)}{2}+n + 60 n/\sqrt{T} +  nN/25\sqrt{T}\Big).$$

As $\sqrt{T}\geq 26$ by constraint \eqref{eqn::constraints}, %
\begin{equation}
\label{eqn::upper_and_lower_TotBound}
    \begin{aligned}
    \frac{n\sqrt{T}}{6}\leq\max
\Big\{
\begin{array}{l}
   \text{total number of men},\\ \hspace*{0.2in}\text{total number of women}
\end{array}
\Big\}
\leq  \rjc{n\sqrt{T}.}
    \end{aligned}
\end{equation}

Let $M=\max\{\text{total number of men}, \text{total number of women} \}$. 
Lemmas~\ref{lem::negative_dependence_two_sex} and~\ref{lem::match rate} yield the following bound on the deviation from the expected number of the number of men in $S$ matched in a given time step:
\begin{equation}
\label{eqn::man_deviation_general_population}
\begin{aligned}
    &\Pr\bigg[\Big|
\begin{array}{l}
    \text{number of men matched}\\
    \hspace*{0.2in}-\Expect[\text{number of men matched}]
\end{array}
    \Big|> \frac{mw\epsilon}{M}\bigg] \leq 2e^{-{mw\epsilon^2}/{3M}}.
\end{aligned}
\end{equation}
By the lower bound on $M$ provided by  \eqref{eqn::upper_and_lower_TotBound}:
\begin{align*}
    &\Pr\bigg[\Big|
\begin{array}{l}
    \text{number of men matched}\\
    \hspace*{0.1in}-\Expect[\text{number of men matched}]
\end{array}
    \Big|> \frac{6mw\epsilon}{n\sqrt{T}}\bigg] \leq \Pr\bigg[\Big|
\begin{array}{l}
    \text{number of men matched}\\
    \hspace*{0.1in}-\Expect[\text{number of men matched}]
\end{array}
    \Big|> \frac{mw\epsilon}{M}\bigg].
\end{align*}
And by the upper bound on $M$ given by  \eqref{eqn::upper_and_lower_TotBound}, $2e^{-{mw\epsilon^2}/{3M}}\leq
2e^{\rjc{-{mw\epsilon^2}/{3n\sqrt{T}}}}$.

We now apply these two bounds to equation \eqref{eqn::man_deviation_general_population} to obtain:
\begin{align*}
    &\Pr\bigg[\Big|
\begin{array}{l}
    \text{number of men matched}\\
    \hspace*{0.2in}-\Expect[\text{number of men matched}]
\end{array}
    \Big| > \frac{6mw\epsilon}{n\sqrt{T}}\bigg] \leq 2e^{-{2mw\epsilon^2}/{9n\sqrt{T}}}.
\end{align*}

The same reasoning can be applied to the number of women matched in $S$.

We  set $\epsilon=\big[\frac{3n\sqrt{T}}{mw}\ln(4n^{3c+1} (T^2/32 + T/8) N)\big]^{1/2}$. 
By the inductive hypothesis, $m+w\leq 7.5n$, and therefore $mw \leq (15n/4)^2$.  We obtain:
\begin{align*}
& \frac{6mw\epsilon}{n\sqrt{T}} 
 \Big[\frac{3mw\ln(4n^{3c+1} (T^2/32 + T/8) N)}{n\sqrt{T}}\Big]^{1/2}
= \rjc{\frac{45\sqrt{3}}{2}}
\Big[\frac{n\ln(4n^{3c+1} (T^2/32 + T/8) N)}{\sqrt{T}}\Big]^{1/2},\\
&\text{and}\hspace*{0.2in} 2e^{-{2mw\epsilon^2}/{9n\sqrt{T}}} 
 \le \frac{1}{2n^{3c+1} (T^2/32 + T/8) N}.
\end{align*}

\hide{
We obtain:
\begin{align*}
&\Pr\bigg[\Big|
\begin{array}{l}
    \text{number of men matched}\\
    \hspace*{0.2in}-\Expect\big[\text{number of men matched}\big]
\end{array}
\Big|> 6 \Big[\frac{9mw\ln(4n^{2c+1} (T^2/32 + T/8) N)}{2n\sqrt{T}}\Big]^{1/2}\bigg]\\
    &\hspace*{0.4in}\leq \frac{1}{2n^{2c+1} (T^2/32 + T/8) N}.
\end{align*}

By the inductive hypothesis, $m+w\leq 7.5n$. Therefore $mw \leq (15n/4)^2$, and: 
\begin{align*}
&\Pr\bigg[\Big|
\begin{array}{l}
    \text{number of men matched}\\
    \hspace*{0.2in}-\Expect\big[\text{number of men matched}\big]
\end{array}
> \frac{45}{2}
\Big[\frac{9n\ln(4n^{2c+1} (T^2/32 + T/8) N)}{2\sqrt{T}}\Big]^{1/2}\bigg]\\
    &\hspace*{0.4in}\leq \frac{1}{2n^{2c+1} (T^2/32 + T/8) N}.
\end{align*}
}

On adding the bounds for the numbers of men and women, this yields:
\begin{align}
\label{eqn::new_variance}
\Pr\bigg[\big|\sum_{d\in S} X(d,\tau)\big|
\leq \rjc{45\sqrt{3}}\Big[ \frac{n\ln(4n^{3c+1} (T^2/32 + T/8) N)}{\sqrt{T}}\Big]^{1/2}\bigg]
\leq \frac{1}{n^{3c+1} (T^2/32 + T/8) N}.
\end{align}

Recall that there are $N$ strips, at most $n^c$ rounds, and, for each strip, there are at most $(T^2/32 + T/8)$ choices of $l$ and $r$. Therefore, the total failure probability is at most $\frac{1}{n^{2c+1}}$.
\end{proof}

\paragraph{Bound on the contribution of $Y$ to strip $s$.}
As for $X$, we define $Y(d_i, \tau, d_j, \tau')$ to be the portion of $Y(d_i, \tau)$ on diagonal $d_j$ at time $\tau'$. 
\hide{\begin{claim}
\label{clm::distr-of-Y}
For all $\ell$, for all $i<k$, and for all $\tau'\ge \tau$,
$\big|\sum_{j>\ell} Y(d_i,\tau,d_j,\tau')\big| \le \big|\sum_{j>\ell} Y(d_k,\tau,d_j,\tau')\big|$.
\end{claim}
In addition:}
\begin{claim}
\label{clm::bound-on-Y-one-strip}
With probability at least $1 - \frac{1}{n^{2c+1}}$, for all $\tau'\ge \tau$, for every strip $s$,
$\big|\sum_{d_i; d_j\in s} Y(d_i,\tau,d_j,\tau')  \big| \le 2\sqrt{\frac{3n}{2} \ln \left(2 T n^{3c + 1}\right)}$.
\end{claim}
The proof of this claim is similar in spirit to that of
Claim~\ref{clm::bound-on-X-one-strip::B}. We defer it to Appendix \ref{appn::imbalance}.

\paragraph{Remaining $X$ and $Y$ in the market.}
Next, we want to show that after $O(T)$ time the portions of $X$ and $Y$ remaining in the market are small. 
\begin{claim}\label{clm::remain::type::1}
$\frac{e^2 \ln 2}{\log_2 (4/3)} \sqrt{T}(\sqrt{T}+\log_{2}(2n^k))$ time after their creation, there is only a $\frac{1}{2n^k}$ fraction of $X(d,\tau)$ and $Y(d,\tau)$ remaining in the Type $1$ strips.
\end{claim}
\begin{proof}
Consider some $X(d,\tau)$ or $Y(d, \tau)$ generated in a Type $1$ strip. 

We first bound $\sum_{j: d_j \in s} (2A_j - I_j - X_j) / 4R$ for any Type $1$ strip $s$. 
By Theorem \ref{thm::lb-size}, the total size of the population is lower bounded by $(1/3)n\sqrt{T}$. By the inductive hypothesis, any Type $1$ strip $s$ has total size at most $2.6n$. The term $\sum_{j: d_j \in s} (2A_j - I_j - X_j)$ is $2$ times the total number of women in strip $s$. By the inductive hypothesis, the number of women in $s$ is at most $1.3n+n/50\sqrt{T}$. Lemma \ref{lem::match rate} provides the following upper bound on the probability
that a man receives a match in a Type $1$ strip:
\begin{align}
\label{enq::match-rate-bound}    
\sum_{j: d_j \in s} \frac{(2A_j - I_j - X_j)}{4R} \leq \frac{1}{2} \cdot \frac{1.3n+ \frac{n}{50\sqrt{T}}}{\frac{1}{6}n\sqrt{T}}<\frac{4}{\sqrt{T}},\hspace*{0.2in}\text{as ($\sqrt{T}\ge 26$ by constraint~\ref{eqn::constraints})}.
\end{align}

Consider any $X(d,\tau,d',\tau')$. If $d'$ is in a Type $1$ strip then by \eqref{enq::match-rate-bound} in one step at most $\frac{4}{\sqrt{T}}$ of it disperses to some location in the same strip, and at least $1-\frac{4}{\sqrt{T}}$ of it moves down distance one. This implies that in $\sqrt{T}/2$ time a Type $1$ strip loses at least $e^{-2}$ of the $X(d,\tau,d',\tau')$ that had been present within it at time $\tau'$. Let $K_1=e^{2}\ln 2$. By time $\tau'+ K_1\sqrt{T}/2$ at least half of the $X(d,\tau,d',\tau')$ in a Type $1$ strip has moved out of the strip. 

We number the Type $1$ strips from top to bottom. Let $\gamma$ be the distribution of $X(d,\tau)$ (or $Y(d,\tau)$) where $\gamma_i$ is the fraction of $X(d,\tau)$ (or $Y(d,\tau)$) in strip $i$. Recall that there are $\sqrt{T}$ Type $1$ strips. We consider the worst case: the $X(d,\tau)$ starts out in the topmost strip. Define a potential function $\phi(\gamma)=\sum_{i=1}^{\sqrt{T}}\gamma_i\cdot 2^{\sqrt{T} - i + 1}.$ Any fraction of $X(d,\tau)$ that has left the bottommost Type $1$ strip contributes nothing to the potential. The initial potential is $2^{\sqrt{T}}$. Every $K_1\sqrt{T}$ time steps, the potential decreases by at least $1/4$. Therefore, after $\frac{1}{\log_2 (4/3)} K_1\sqrt{T}\log_{2}(2^{\sqrt{T}}2n^k)$ time, the potential would have reduced to at most $\frac{1}{2n^k}$, which means that the fraction of $X(d,\tau)$ (or $Y(d,\tau)$) in the Type $1$ strips after $\frac{1}{\log_2 (4/3)} K_1\sqrt{T}(\sqrt{T}+\log_{2}(2n^k))$ time is at most $\frac{1}{2n^k}$.
\end{proof}

We will analyze the progress through the Type 2 strips, apart from the bottommost one, in a similar way. The proof can be found in Appendix \ref{appn::imbalance}.

\begin{claim}\label{clm::remain::type::2}
$\frac{e^2 \ln 2}{\log_2 (4/3)} \sqrt{T}(\sqrt{T}+\log_{2}(2n^k)) + \frac{e^{12} \ln 2}{4 \log_2 (4/3)} T\log_{2}(2n^k\sqrt{T})$ time after their creation, there is only $\frac{1}{n^k}$ fraction of $X(d,\tau)$ and $Y(d,\tau)$ remaining in any strip other than the bottommost Type 2 strip.
\end{claim}

\paragraph{The Total Bound on Imbalance}
Now we can bound the total imbalance in a strip $s$ at time $\tau'$. Let $\kappa = \frac{e^2 \ln 2}{\log_2 (4/3)} \sqrt{T}(\sqrt{T}+\log_{2}(2n^k)) + \frac{e^{12} \ln 2}{4 \log_2 (4/3)} T\log_{2}(2n^k\sqrt{T})$.
We divide the time interval $[0, \tau']$ into two periods: $\left[0, \tau' - {\kappa}\right]$ and  $\left[\tau' - {\kappa} + 1, \tau'\right]$. 
\begin{itemize}
    \item In the first period, we bound each $|X(d,\tau)|$ and $|Y(d,\tau)|$ by $7.5n$ as no strip can have more than $7.5n$ agents on it by Lemma~\ref{lem::ind-bound}.  By Claims~\ref{clm::remain::type::1} and \ref{clm::remain::type::2}, the total imbalance for this period is at most  $15n T n^c / n^k$;
    \item For the second period, using Claims \ref{clm::bound-on-X-one-strip}, \ref{clm::bound-on-X-one-strip::B}, and \ref{clm::bound-on-Y-one-strip}, The total imbalance is at most $\ceil{\kappa} \cdot \Big(192\sqrt{ \frac{n\ln(4n^{3c+1} (T^2/32 + T/8) N)}{\sqrt{T}}} +2\sqrt{\frac{3n}{2} \ln \left(2 T n^{3c + 1}\right)} \Big) $.
\end{itemize} 

We choose $k = c + 4$ and sum them up. We desire that both these contributions to the imbalance add up to no more than $n/25\sqrt{T}$. Using $n\geq T\geq 676$ (by the constraint \eqref{eqn::constraints}), we simplify this condition to conclude that it suffices to have:
$$n \geq (3654 + 2436e^{12} + 546(e^{12} + 1) c)^2(3c+4) T^3 (\log_2 n)^2 \ln n.$$

The details of this calculation can be found in Appendix \ref{appn::imbalance}.

Finally, the failure probability of $2/n^{2c+1}$ arises from Claims~\ref{clm::bound-on-X-one-strip::B} and~\ref{clm::bound-on-Y-one-strip}, which each have failure probability at most $1/n^{2c+1}$.

\end{proof}

%% file: ind-hyp-holds-init.tex
\subsubsection{Initialization}

\label{sec::initialization}
 \begin{theorem}
 \label{thm::initialization}
Suppose that constraint~\eqref{eqn::constraints} holds. If all agents follow the modified reasonable strategy, then $H(\sqrt{T})$ holds with probability at least $1-\frac{1}{n^{c+1}}$.
 \end{theorem}
The proof is similar to the earlier analysis and can be found in Appendix \ref{appn::init}.

%% file: numerics.tex
\section{Numerical Simulations}
\label{sec::simulations}

We have demonstrated a strategy which is asymptotically close to optimal with regard to minimizing the average loss experienced by agents. Complementing this, in this section we simulate the  evolution of the system for moderately large values of $n$ and $T$. In order to gain a sense of the overall stability of the system, we track the total population over time.

We now discuss some observations based on our simulations.\footnote{For every pair of $n$ and $T$ that we considered in the discrete setting, we ran the simulation $10$ times; letting each run for $2000$ iterations. The error ranges mentioned below are obtained from the range of values we obtained over these $10$ runs. The values for each run can be found in Appendix \ref{appn::data}.}

\begin{figure}[thb]

\centering

\begin{minipage}[b]{0.45\textwidth}

\centering

\includegraphics[width=1\textwidth]{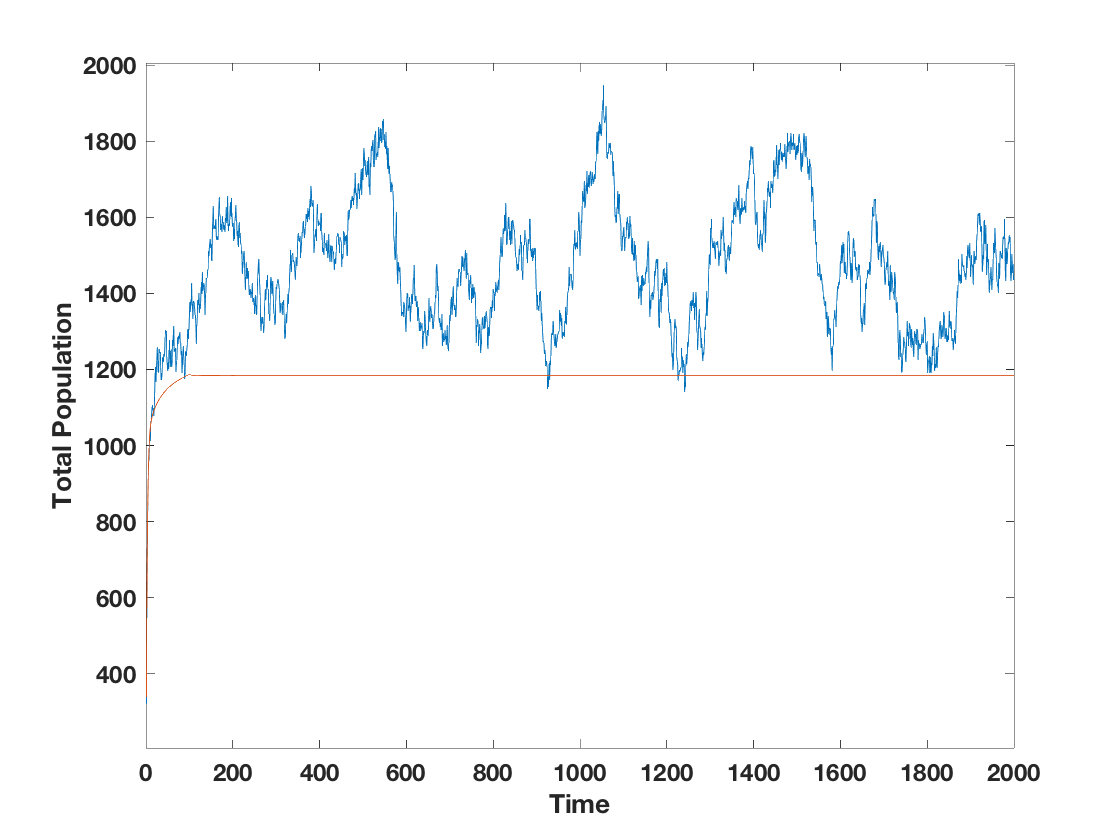}
\caption{\label{fig:one}The evolution of total population over time in the discrete (blue) and continuum (orange) settings for $n=500$, $T=100$, using the reasonable strategy.}
\end{minipage}
\hspace{0.08\textwidth}
\begin{minipage}[b]{0.45\textwidth}
\centering
\includegraphics[width=1\textwidth]{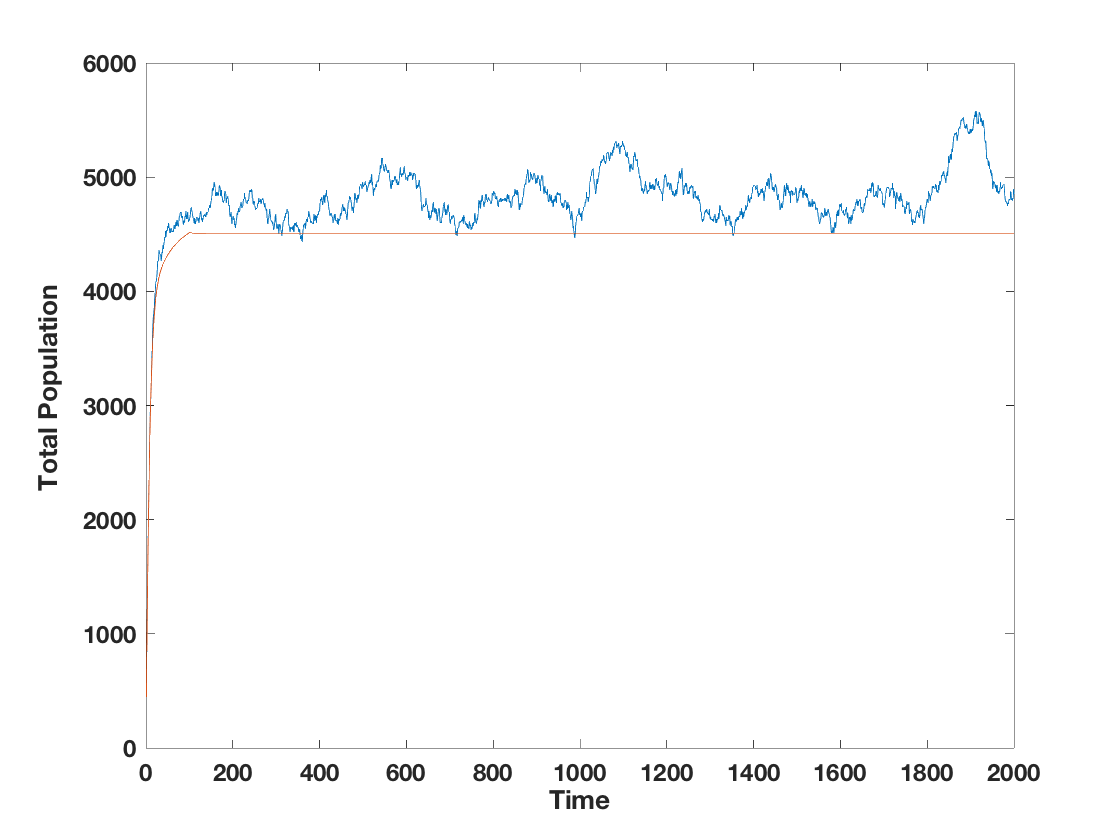}
\caption{\label{fig:two}The evolution of total population over time in the discrete (blue) and continuum (orange) settings for $n=500$, $T=100$, using the modified reasonable strategy.}
\end{minipage}
\end{figure}

For the continuum model we obtain reasonably rapid convergence---in about $T$ time---whereas for the discrete model in a similar time the system reaches its long-term average value, but with somewhat chaotic oscillations about this value, as shown in Figures~\ref{fig:one} and~\ref{fig:two}. In addition, the long-term average population for the discrete case is a bit larger than the continuum equilibrium value. (This is not surprising, for both variance and male/female imbalances will reduce the match rate.)

For moderate values of $n$ and $T$, the average loss in the modified reasonable strategy is better than the asymptotic bound we obtain. For example, consider the $n=500$, $T=100$ case. We prove an upper bound on the total loss of $11T\sqrt{T}$ but the simulation achieves an average loss of just $2.21T\sqrt{T}(\pm 1.7\%)$. The total populations are significantly closer (an upper bound of close to $1.5nN$ in our theorem vs.\ close to $n\sqrt{T}$ in the simulation).

For the case where agents use the modified reasonable strategy, we also examine the average population size and average loss for various $T$ (with $n$ fixed at $500$). Figure \ref{fig:three} shows a plot of $\text{Average Population}/n\sqrt{T}$ and $\text{Average Loss}/\sqrt{T}$ for five different values of $n$. The result is quite consistent with the $n\sqrt{T}$ scaling of the average total size and the $\sqrt{T}$ scaling of the average loss that we prove hold asymptotically,
even though these are only moderately large values of $n$ and $T$, and even though we are not in the $n$ much greater than $T$ regime of our analysis.

\begin{figure}[htb]

\centering

\begin{minipage}[tb]{0.45\textwidth}

\centering

\includegraphics[width=1\textwidth]{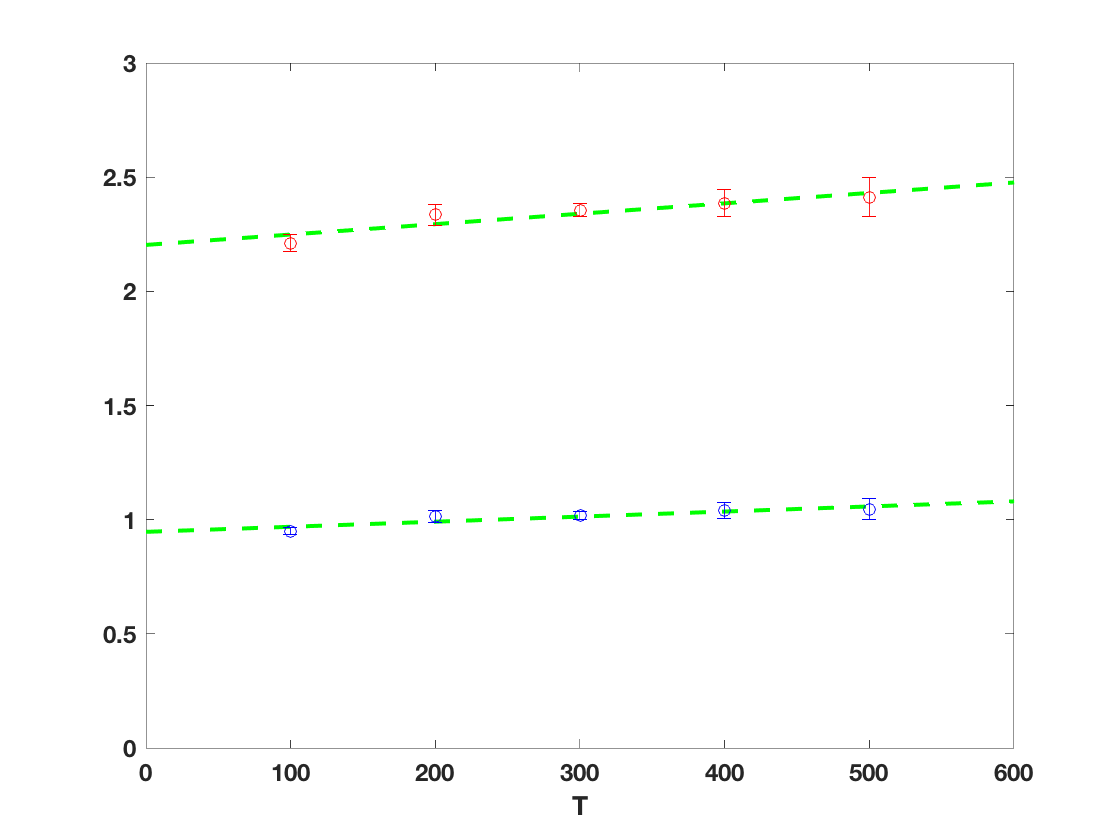}
\caption{\label{fig:three} $\text{Average Population}/n\sqrt{T}$ (blue) and $\text{Average Loss}/\sqrt{T}$ (red) for five  values of $T$, using the modified reasonable strategy.}
\end{minipage}
\hspace{0.08\textwidth}
\begin{minipage}[tb]{0.45\textwidth}
\centering
\includegraphics[width=1\textwidth]{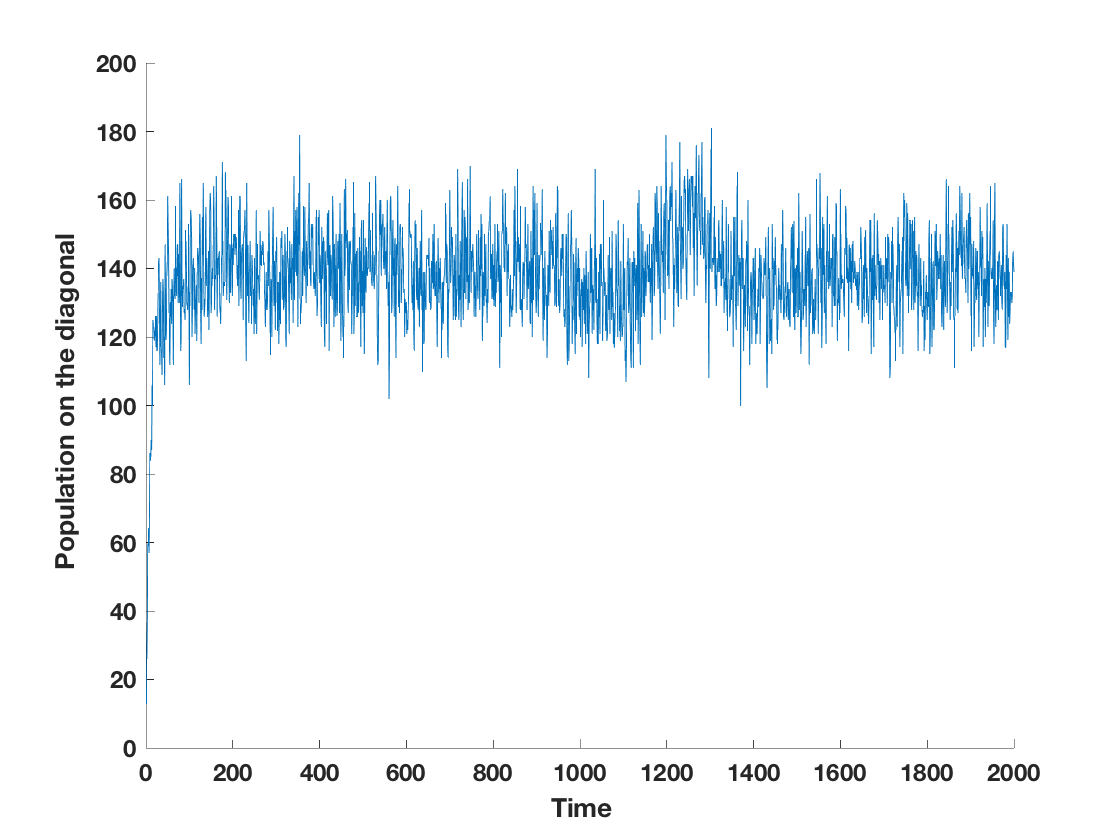}
\caption{\label{fig:four}The evolution of the population on a "typical diagonal" over time in the discrete setting (for $n=500$, $T=100$), using the modified reasonable strategy.}
\end{minipage}

\end{figure}

Finally, we examine the population on a "typical" diagonal\footnote{Here we consider the diagonal region with width $1$ that starts at value $3T/2$ at the top ($t=0$) boundary.} in the case that all agents are following the modified strip strategy (Figure~\ref{fig:four}). Notice that the range of oscillations for this value is large compared to the range of oscillations in the total population (Figure~\ref{fig:two}). Furthermore, these 
oscillations proceed at a much faster rate than
the changes in the overall population.
The results also indicate that while the system
remains within reasonable bounds, there is
substantial ongoing variation, particularly at a local level.

%% file: open-problems.tex
\section{Open Problems}
\label{sec:open-problems}

Two natural extensions of our model come to mind.
\begin{itemize}
    \item The values in our model are common to all agents, but in reality agents will have individual preferences. This could be captured with a model in which each agent $a$ has a value for agent $b$ given by $v_a + w_{a,b}$, where $v_a$ is a common public value while $w_{a,b}$ is an idiosyncratic private value of $a$ for $b$. The combining of public and private values has been studied in the literature on matchings in other settings~\cite{ashlagi2020clearing, lee2016incentive}.
    
    
    \item In our model, agents receive match proposals that are generated by choosing agents from the other side of the market uniformly at random. It would be interesting to consider a more sophisticated method of recommending matches, with recommended matches being localised in value and time around the agent.
\end{itemize}

Another intriguing direction concerns the stability of this system. We have shown that if the agents play the modified reasonable strategy then with high probability the strip sizes, the total size, and the imbalance between men and women in any strip, all remain within some range. But we conjecture that if any of these parameters were to have a large deviation which took it outside  its typical range, then with high probability it would soon return
to being within this range.



%% file: appendix.tex
\appendix

\input{appendix_prelim}

\input{appendix_lower_bound}
\input{appendix_upper_bound}
\input{appendix_data}

%% file: appendix_prelim.tex
\section{Deferred Proofs from Section \ref{sec::prelim}}
\label{appn::prelim}
\begin{proof} (Of Lemma~\ref{lem::negative_dependence_two_sex})~
Let $N=\max\{N_1,N_2\}.$
Consider any subset $S\subseteq[n]$ where $|S|=k$. W.l.o.g.\ let $S=[k]$. Then,
\begin{align*}
E\Big[\prod_{i\in S} X_i\Big] &= \chance\Big[{\prod_{i\in S} X_i=1}\Big]=\pr{X_1=1,X_2=1,\ldots, X_k=1}\Big]\\
&=\pr{X_1=1}\cdot \pr{X_2=1|X_1=1},\ldots ,\pr{X_k=1|X_1=1,X_2=1,\ldots, X_{k-1}=1}\\
&=\frac{r}{N}\cdot\frac{r-1}{N-1},\ldots,\frac{r-k+1}{N-k+1}.
\end{align*}
Hence,
$$E\Big[\prod_{i\in S} X_i\Big]\leq \Big(\frac{r}{N}\Big)^k,$$
while
$$\prod_{i\in S} E[X_i]=\Big(\frac{r}{N}\Big)^k.$$
Similarly,
\begin{align*}
E\Big[\prod_{i\in S} (1-X_i)\Big] &= \chance\Big[{\prod_{i\in S}(1-X_i)=1}\Big]=\pr{X_1=0,X_2=0,\ldots, X_k=0}\Big]\\
&=\pr{X_1=0}\cdot \pr{X_2=0|X_1=0},\ldots ,\pr{X_k=0|X_1=0,X_2=0,\ldots, X_{k-1}=0}\\
&=\frac{N-r}{N}\cdot\frac{N-r-1}{N-1},\ldots,\frac{N-r-k+1}{N-k+1}.
\end{align*}
Hence,
$$E\Big[\prod_{i\in S}(1-X_i)\Big]\leq \Big(\frac{N-r}{N}\Big)^k,$$
while
$$\prod_{i\in S} E[1-X_i]=\Big(\frac{N-r}{N}\Big)^k.$$
Thus the set $\{X_i\}$ is negative cylinder dependent.
By \cite[Theorem 3.4]{PanconesiSrinivasan} with $\lambda=1$, Chernoff bounds for sums of independent random variables apply to the sums of negative cylinder dependent random variables as well.
This concludes the proof.
\end{proof}

\begin{proof} (Of Lemma~\ref{lem::chernoff_bound})~

First we prove that 
\begin{align*}
    \chance\left[\sum X_i \geq (1 + \delta) \overline{\mu} \right] \leq e^{- \frac{\delta^2 \overline{\mu}}{3}}.
\end{align*}
Let $\overline{\mu}=(1+\theta)\mu$. By Lemma \ref{lem::negative_dependence_two_sex},
\begin{equation}
\label{eqn::chernoff_with_upper_bounded_mean}
\begin{aligned}
    \chance\left[\sum X_i \geq (1 + \gamma) \mu \right] \leq e^{- \frac{\gamma^2 \mu}{3}}.
\end{aligned}
\end{equation}
Set $(1+\gamma)= (1+\theta)(1+\delta)$. Then, from equation \eqref{eqn::chernoff_with_upper_bounded_mean} we obtain,

\begin{equation*}
\begin{aligned}
    \chance\left[\sum X_i \geq (1 + \delta) \overline{\mu} \right] &\leq \exp\bigg[- \frac{(\theta+\delta+\theta\cdot\delta)^2 \mu}{3}\bigg]\leq \exp\bigg[- \frac{(\delta+\theta\cdot\delta)^2 \mu}{3}\bigg]\\
    &\leq \exp\bigg[- \frac{\delta^2(1+\theta)^2 \mu}{3}\bigg]\leq \exp\bigg[- \frac{\delta^2(1+\theta) \mu}{3}\bigg]\leq \exp\bigg[- \frac{\delta^2 \overline{\mu}}{3}\bigg]
\end{aligned}
\end{equation*}
which proves the claim.
We will now prove that
\begin{align*}
    \chance\left[\sum X_i \leq \mu - \delta \overline{\mu} \right] \leq e^{- \frac{\delta^2 \mu}{2}}
\end{align*}

Let $\theta=\frac{\mu}{\overline{\mu}}$. 
By Lemma \ref{lem::negative_dependence_two_sex},
\begin{align*}
    \chance\left[\sum X_i \leq \mu - \delta \mu \right] \leq e^{- \frac{\delta^2 \mu}{2}}
\end{align*}

Let $\gamma=\frac{\overline{\mu}}{{\mu}}$. Note that $\gamma\geq1$

\begin{equation*}
\begin{aligned}
    \chance\left[\sum X_i \leq \mu-\delta\overline{\mu} \right]&= \chance\left[\sum X_i \leq \mu-\gamma\delta\mu \right] \leq \exp\bigg[- \frac{(\gamma\delta)^2 \mu}{2}\bigg]\leq \exp\bigg[- \frac{\delta^2\gamma\overline{\mu}}{2}\bigg]\leq \exp\bigg[-\frac{\delta^2\overline{\mu}}{2}\bigg]\\
\end{aligned}
\end{equation*}
which completes the proof.

\end{proof}

%% file: appendix_lower_bound.tex
\section{Deferred Calculation from Section 5}
\label{appn::loss_prob_calculation}
\begin{proof} (Final calculation for Theorem~\ref{thm::lb-loss-two_sex})~
\begin{align*}
\chance[{\mathcal E}] \cdot\left(1-\tau T  e^{-\frac{\delta^2 nw^2}{6T^2}}-\tau e^{-\frac{\epsilon^2 n}{4}}\right)
& \ge \Big(1-e^{-\frac{\delta^2nw}{6T}}\Big)^{\tau\frac{T}{w}}\cdot\Big(1-\tau T e^{-\frac{\delta^2nw^2}{6T^2}}-\tau e^{-\frac{\epsilon^2 n}{4}}\Big)\\
&\geq \Big[1-\Big(\frac{1}{3n^cT\tau}\Big)^{4\sqrt{T}}\Big]^{4\tau\sqrt{T}}\cdot\Big(1-\frac{1}{3n^c}-\frac{1}{3n^{c}}\Big)\\
&\geq \Big[1 - 4\tau\sqrt{T}\Big(\frac{1}{3n^cT\tau}\Big)^{2} \Big]  \cdot\Big(1-\frac{1}{3n^c}-\frac{1}{3n^{c}}\Big) \\
&\geq \Big[1 - \Big(\frac{1}{3n^c}\Big)^{4 \sqrt{T}}\Big]\cdot\Big(1-\frac{2}{3n^c}\Big)
\geq \Big( 1  - \frac{1}{n^c}\Big).
\end{align*}
\end{proof}

%% file: appendix_upper_bound.tex
\section{Deferred Proofs from Section \ref{sec::upper_bound}}

\begin{proof} (Of Lemma~\ref{lem::ind-bound}.)~
We complete the sketch proof by bounding the failure probability. Per time step,
Theorems~\ref{thm::lb-size}, \ref{thm::total_size_upper_bound}, \ref{type1_upper}, \ref{type2_upper}
all have failure probability of at most $1/n^{2c+1}$, Theorem~\ref{thm::imbalance_bound} has
failure probability at most $2/n^{2c+1}$,
and Theorem~\ref{thm::initialization}, which is
applied once,
has failure probability at most $1/n^{c+1}$.
Theorem~\ref{thm::last_strip_size_upper_bound}
does not introduce any additional possibility of
failure.
Multiplying by the $n^c$ possible time steps,
gives a  total failure probability of at most
$7/n^{c+1} < 1/n^c$.
\end{proof}

\subsection{Upper Bound on Loss due to a Match}
\label{appn::loss_from_match}
\begin{proof} (of Lemma~\ref{lem::loss_in_strip}).
Consider an agent (Agent $1$) at value $v$ and time $t$. Suppose they match with another agent (Agent $2$) who is present in the same strip. 
The worst location for Agent 2 is to be on the low value strip boundary, and on this boundary to be at 
one of the endpoints. 

\noindent
\emph{Type $1$ strip}.
If Agent $2$ is at the top endpoint, 
Agent $1$ obtains utility $w\cdot T-t$, where $w$ is the value at the top endpoint.
We can see that $w\ge v- \sqrt{T} -2t$ (move from $v$ horizontally to the lower boundary, a distance of at most $\sqrt{T}$ and then move up to the $t=0$ location, which subtracts $2t$ from the value).
Thus the utility Agent $1$ receives is at least
$(T-t)(v- \sqrt{T} -2t)$. Therefore the loss is at most $t(v+2T+\sqrt{T}) \le 4tT + 2T\sqrt{T}$.

If Agent $2$ is at the lower endpoint of a Type 1 strip, we argue as follows.
The $v\cdot T-t$ product is equal at the two endpoints of a boundary,
and therefore the loss is greatest at the top endpoint, 
for the utility garnered by Agent $1$ would be $w\cdot (T-t)$ and not $w\cdot T$,
whereas at the bottom endpoint the garnered utility is $2T\cdot w/2 =wT$.

\smallskip
\noindent
\emph{Type $2$ Strip}.
We define the following $t$ values:
$a$ is the  value for the left end of the top boundary of the strip,
$b$ the value for the left end of the bottom boundary.
and $c$ the value for the right end of the bottom boundary,
Then $a\le (2t+T-v)/2$, $b \leq 2a+\sqrt{T}$, and $c = T/2 + b \leq T/2 + 2a + \sqrt{T}$.

If Agent $2$ is at the lower endpoint, then
Agent 1 would receive utility $2T\cdot (T-c) \ge 2T(v -T/2 -2t -\sqrt{T})$.
Thus the loss is at most $T(T-v) + 4tT + 2T\sqrt{T}
 \le 4tT + 2T\sqrt{T}$.

If Agent $2$ is at the top endpoint and Agent 1 is older than Agent 2,
then Agent $1$ receives utility $T(T-t)$. 
As we are in a Type 2 strip, $(v-T)\le 2t$ or $v\le T+2t$. 
So Agent $1$ incurs a loss of at most $vT-T(T-t) \le (T+2t)T - T(T-t) \le 3tT$.

While if Agent $1$ is no older than Agent $2$,
then Agent $1$ receives utility $T(T-b)\ge T(v - 2t - \sqrt{T})$. 
Thus the loss is at most $vT - T(v - 2t - \sqrt{T}) = 2Tt + T\sqrt{T}$.
\end{proof}

\subsection{Lower Bound on the Total Population}
\label{appn::lower_bound_on size}
\begin{proof} (of Theorem~\ref{thm::lb-size})
The agents enter with one of $T$ values chosen uniformly at random and are equally likely to be men or women. Hence, for all $n^c$ time steps, for each value $v$,
$$\chance\Big[\text{At most $\frac{n(1+\epsilon)}{2T}$ men enter with value $v$}\Big]\geq 1- n^c Te^{-\frac{\epsilon^2n}{6T}}.$$
Call this event $\mathcal E$. Henceforth we condition on $\mathcal E$.

Let's consider those agents that enter at times  in the range $[\tau-\sqrt{T}+1,\tau]$ for some $\tau \le n^c$. We want to lower bound the number of these agents who are present in the pool for the match at time $\tau$.

In fact, henceforth, We will only consider men with values in the range $[T+\sqrt{T},2T)$. Among these men, consider those who have been in the pool for $t$ time, where $0\le t < \sqrt{T}$.
Let $p_i^t$ be the probability that during their  $t$th time step, the men in strip $i$ are offered a match in their own strip. Even if all these men were still present in the matching pool,
\begin{align*}
\chance\Big[\text{\# of these men matched in strip $i$ at age $t$ } \le \frac{n(1+\delta)(1+\epsilon)\cdot p_i\cdot w_i}{2T} \Big]\geq 1-e^{-\frac{\delta^2np_iw_i}{6T}},
\end{align*}
where $w_i$ is the horizontal width of strip $i$ occupied by these men when aged $t$.
For every Type 1 strip, $w_i \le \sqrt{T}$.
For the one Type 2 strip, since all values are
at least $T+\sqrt{T}$, 
for ages up to $\sqrt{T}$, $w_i \le \sqrt{T}$.  
By applying $\overline{\mu} = \frac{n(1 + \epsilon) \max\{p_i, \frac{1}{T}\} \sqrt{T}}{2 T}$ in Lemma~\ref{lem::chernoff_bound}, it follows that:
\begin{align*}
\chance\Big[\text{\# of these men matched in strip $i$ at age $t$ } \le \frac{n(1+\delta)(1+\epsilon)\cdot \max\{p_i, \frac{1}{T}\}}{2\sqrt{T}} \Big]&\geq 1-e^{-\frac{\delta^2n \max\{p_i, \frac{1}{T}\} \sqrt{T}}{6T}} \\
&\geq 1-e^{-\frac{\delta^2n}{6T^{1.5}}},
\end{align*}

The sum of the match probabilities---the $p_i$s--- is at most $1$.  Notice that at any fixed time we only need to consider $\sqrt{T}$ strips, because at any time step, the men we are considering will occupy only $\sqrt{T}$ many strips. This implies $\sum \max\{p_i, \frac{1}{T}\} \leq 1 + \frac{1}{\sqrt{T}}$. Therefore, 

$$\chance\Big[\begin{array}{l}\text{\# of these men being matched} \\ \text{over all the strips at age $t$}\end{array} \le \frac{ (1+\frac{1}{\sqrt{T}})n(1+\delta)(1+\epsilon)}{2\sqrt{T}}\Big]\geq 1-\sqrt{T}\cdot e^{-\frac{\delta^2n}{6T^{1.5}}}.$$

Hence, we can bound the probability of the number of men who entered at time $\tau -\Delta+1$ and left by time  $\tau$, for any $\Delta \le \sqrt{T}$, as follows:
$$\chance\Big[\begin{array}{l}\text{\# men being matched} \\ \text{in their first $\Delta$ steps}\end{array} \le \frac{ (1+\frac{1}{\sqrt{T}})n\Delta(1+\delta)(1+\epsilon)}{2\sqrt{T}}\Big]\geq 1- \Delta T\cdot e^{-\frac{\delta^2n}{6T^{1.5}}}.$$

Consequently, we can bound the probability for the number of men that enter in the time interval $[\tau -\sqrt{T}+1,\tau-1]$ and are matched no later than time $\tau-1$ as follows, where
we sum over all $\tau \le n^c$:
\begin{equation*}
    \begin{aligned}
    &\chance\bigg[\begin{array}{l}\text{\# men who entered and were} \\  \text{matched in a $\sqrt{T}-1$ time window}\end{array} \le \frac{(1+\frac{1}{\sqrt{T}})n(1+\delta)(1+\epsilon)(\sqrt{T}-1)}{4} \bigg] \geq 1- \frac 12 n^c T^{1.5} e^{-\frac{\delta^2n}{6T^{1.5}}}.
    \end{aligned}
\end{equation*}

We set $\delta=\Big[\frac{6T^{1.5}}{n} \ln \big(10n^{2c+1}n^c T^{1.5}\big)\Big]^{1/2}$ 
and $\epsilon=\Big[\frac{6T}{n}\ln\big(20n^{2c+1}n^c T\big)\Big]^{1/2}$. By constraint~\eqref{eqn::constraints}, $c\geq 1$, $400 \leq T\leq n$, and $n\geq 96T^2(3c+3)\ln n$, 
therefore $\delta\leq 1/4$ and $\epsilon\leq  {1}/{64}$.  
This yields the bound:
\begin{align*}
\chance\Big[\begin{array}{l}\text{\# men who entered in a $\sqrt{T}-1$} \\ \text{ window being matched }\end{array} \ge \frac{\frac{65}{64} \cdot \frac 54 n\sqrt{T}} {4} \Big]
\geq 1-\frac{1}{20n^{2c+1}}.
\end{align*}

Since we have been conditioning on $\mathcal E$, this bound holds with probability at least $1-\frac{1}{10n^{2c+1}}$.
The same bound applies to the women.

Recalling that we excluded the men with values less than $T+\sqrt{T}$,
this yields the following lower bound on the total population size, throughout this $n^c$ time period:
$$n\sqrt{T} - \frac{n(1 + \epsilon)}{2\sqrt{T}} - 0.635n\sqrt{T}\geq \frac 13 n\sqrt{T},$$
with probability at least $1-\frac{1}{5n^{2c+1}}.$
\end{proof}

\subsection{Upper Bound on The Total Population}
\label{appn::total_upper}

\begin{proof} (Of Theorem~\ref{thm::total_size_upper_bound}.)~
Let $P(t)$ be the total population at the start of time step $t$. Let $N$ be the total number of strips. By Constraint~\ref{eqn::constraints}, $T\geq 676$, so $N\leq \sqrt{T}+\log_2 \sqrt{T} + 1 \leq 5\sqrt{T}/4.$

If $P(t)\le \frac{3}{2}nN$, then $P(t+1) \le \frac{3}{2}nN+n$.
So we will only consider the case that $P(t)>\frac{3}{2}nN$.
In this case, the average strip population at the start of step $t$ is more than $\frac{3}{2}n$. 

Next, we upper bound the number of men in the population; the same bound applies to the number of women. 
\hide{\rjc{Using the bound on the strip imbalance from $H(t)$, we can see it is at most $[P(t) + nN/25\sqrt{T}]/2 =P(t)/2 + nN/50\sqrt{T}$.
WLOG we will assume that there are at least as many men as women in the overall population.
}}

\hide{Let
$$L\coloneqq\Big\{i \,|\, s_i \geq\frac{n}{\sqrt{T}}\Big\}.$$}

By $H(t)$, clause~\ref{itm::pop-imb},
the excess of men over women in each strip is
at most $n/25\sqrt{T}$ except for the last Type $2$ strip.
So the excess over all these $N - 1$ strips is at most $n(N-1)/25\sqrt{T}$. For the last Type $2$ strip, the population is less than $60n / \sqrt{T}$ which is smaller than $40 P(t) / 676$ as $T \geq 676$.
Consequently, there are at most $P(t)/2 + nN/50\sqrt{T}  + 20 P(t) / 676 \le 11P(t)/20$ 
men in the total population.

The expected number of matches in strip $i$, $\mu_i$, is given by 
$$E[\mu_i] =\frac{(\text{ number of women in strip }i) \times (\text{ number of men in strip } i)}{\text{number of men in the whole population}}.$$ Let $s_i$ denote the population of the $i$-th  strip.
The denominator is at most $\tfrac{11}{20} P(t)$ and at least $\tfrac{1}{2} P(t)$. The numerator is minimized when the number of women and men in the strip are as far apart as possible. So, for the strips other than the last Type $2$ strip, the numerator is at least 
$(s_i/2 + n/50\sqrt{T})(s_i/2 - n/50\sqrt{T}) =
s_i^2/4 - n^2/2500T^2$.
The numerator is maximized when the numbers of women and men are equal. Therefore,
\begin{align}
\label{eqn::bounds-on-Pt}
   \frac{s_i^2/4 - n^2/2500T^2}{\frac{11}{20}  P(t)}\leq E[\mu_i] \leq  \frac{s_i^2}{2 P(t)}. 
\end{align}

Consider an indicator random variable $X_i$ for each man in this strip, which is $1$ if that man gets matched. By Lemma \ref{lem::negative_dependence_two_sex} we can use a Chernoff bound to obtain:
\begin{align}\chance\Big[\text{number of matches in strip } i \leq \mu_i(1-\epsilon)\Big]\leq  e^{-\epsilon^2\mu_i/2}. \label{ineq::chernoff_upper_total}
\end{align}

For $E[\mu_i]  \ge \alpha n / \sqrt{T}$, 
\begin{align*}
    \chance\big[\text{number of matches in strip } i \leq E[\mu_i] (1-\epsilon)\big]\leq  e^{-\epsilon^2E[\mu_i] /2} \leq e^{-\epsilon^2 \alpha n / (2\sqrt{T})}. 
\end{align*}
Let $\epsilon = \sqrt{\frac{2\sqrt{T}}{\alpha n}\ln (N n^{2c + 1}) }$. Then $\chance\big[\text{number of matches in strip } i \leq \mu_i(1-\epsilon)\big] \leq \frac{1}{N n^{2c + 1}}$

For $E[\mu_i]  < \alpha n / \sqrt{T}$, let $\epsilon = \frac{\theta}{E[\mu_i] }$, then \eqref{ineq::chernoff_upper_total} becomes 
\begin{align*}
    \chance\big[\text{number of matches in strip } i \leq E[\mu_i]  - \theta \big]\leq  e^{-\theta^2/(2 E[\mu_i] )} \leq e^{-\theta^2 \sqrt{T}/(2 \cdot \alpha n )}.
\end{align*}
Let $\theta = \sqrt{\frac{2 \alpha n}{\sqrt{T}} \ln (N n^{2c + 1})} = \alpha \epsilon \frac{n}{\sqrt{T}}$, then $\chance\big[\text{number of matches in strip } i \leq \mu_i - \theta \big]\leq  \frac{1}{N n^{2c + 1}}$ 

Let $\texttt{NL}$ be the set of all strips except for the last Type $2$ strip. Then, with probability at least $1 - \frac{1}{n^{2c+1}}$, the number of matches is larger than $(1 - \epsilon) \sum_{i \in \texttt{NL}}  E[\mu_i]  - N \theta$. In addition, by \eqref{eqn::bounds-on-Pt}, $\sum_{i \in \texttt{NL}}  E[\mu_i] $ is lower bounded by:
\begin{align*}
\sum_{i \in \texttt{NL}} \frac{s_i^2/4 - n^2/2500T^2}{\tfrac {11}{20}  P(t)} &\geq  \sum_{i \in \texttt{NL}} \frac{\frac{1}{4} \left(\frac{9 P(t)}{10N}\right)^2 - n^2 /2500T^2}{\tfrac {11}{20}  P(t)} \\
    &\geq \left(\frac{81}{220} \frac{P(t)}{N} - \frac{n^2 N}{1375 T^2 P(t)}\right)  \\
    & \geq  \left(\frac{243}{440} n - \frac{2n}{4125 T^2}\right)  \\
    & \geq  \frac{11}{20} n. 
\end{align*}
The first inequality follows as $\sum_{i \in \texttt{NL}} s_i \ge \frac{9}{10}P(t)$, and the next to last inequality follows as $P(t) \geq \frac{3}{2}n N$. Let $\alpha = 0.4$. Since $\epsilon \leq \frac{1}{22}$, as $n \geq 2420 \sqrt{T} (2c + 2) \ln n$, and $N \theta \leq \frac{5}{4} \alpha \epsilon n$, as $N \leq \frac{5}{4}\sqrt{T}$ and $\theta = \alpha \epsilon n / \sqrt{T}$,   
\begin{align*}
    (1 - \epsilon) \sum_{i \in \texttt{NL}} E[\mu_i]  - N \theta \geq (1 - \epsilon) \frac{11}{20} n - \frac{5}{4} \alpha \epsilon n \geq \frac{n}{2}.
\end{align*}
This means the total number of people matched in the market is greater than $n$, which is the number of people entering, which completes the proof.

\end{proof}

\subsection{Upper Bound on the Size of a Type 1 Strip.}
\label{appn::type1_upper}

\begin{proof} (Of Theorem~\ref{type1_upper})~
Consider any Type $1$ strip $s$. For any two points $(v,t_1)$ and $(v,t_2)$ in $s$ which have the same value $v$, we have $|t_2-t_1|\leq \frac{\sqrt{T}}{2}$.
Let $s'$ be the strip immediately to the right of $s$.

We are going to lower bound the number of matches in time step $t$. Let's consider the agents who will be in strip $s$ at time $t+1$. They will all enter $s$ during a length $\sqrt{T}/2$ time interval ending at time $t+1$. 

Let $P_{t-\sqrt{T}+1}$ be the agents in strip $s'$ at time $t-\sqrt{T}$.
By the inductive hypothesis, applied to $s'$ at time $t-\sqrt{T}$, we know that $|P_{t-\sqrt{T}+1}| \le dn$.
We are going to track the subset of these agents who remain
in the system after each step of matches, for the next $\sqrt{T}$ steps, along with the new agents who join the diagonals used
by this subset of agents.
Define $S_{t-\sqrt{T}+i}$ to be the rightmost $\sqrt{T}+1 -i$ diagonals in $s'$ plus the leftmost $i-1$ diagonals in $s$,
for $1\le i \le \sqrt{T}+1$.
Let $P_{t-\sqrt{T}+i}$ be the population occupying $S_{t-\sqrt{T}+i}$
at the start of step $t-\sqrt{T}+i$.
Then $P_{t-\sqrt{T}+i+1}$ is obtained from $P_{t-\sqrt{T}+i}$ by removed matched agents and then adding the new agents
for the diagonals in $S_{t-\sqrt{T}+i+1}$.
Our analysis will show that, with high probability,
for each of these $\sqrt T$ steps,
the number of matches is at least the number of new agents.
This implies the upper bound on the strip population
continues to hold.

By means of a Chernoff bound, we observe that the number of new agents per step can be bounded with high probability as follows:
\begin{equation}
\label{eqn::type1_entering}
    \begin{aligned}
    \pr{\text{\# new agents } \leq \frac{n}{\sqrt{T}}(1+\delta)}  \ge 1-e^{-\frac{n\delta^2}{3\sqrt{T}}}.
    \end{aligned}
\end{equation}
Let $\delta = \sqrt{\frac{3\sqrt{T}}{n} \ln (N n^{2c + 1})}$.
As $ n\ge 60 T (2c+2) \ln n$, $\delta \leq \frac{1}{20}$,
which
yields 
\begin{align}
\label{eqn::type-one-new-agent-bound}
\chance\Big[\text{\# new agents} \leq \frac{41n}{40\sqrt{T}} =1.025\frac{n}{\sqrt{T}}\Big]  \ge 1-\frac{1}{N n^{2c + 1}}.
\end{align}

By \eqref{eqn::upper_and_lower_TotBound},
the maximum of the number of men and number of women in the market is at most $n\sqrt{T}$.
\hide{
Now, we will bound the number of agents matched in a single step. We assume that at the start of the round, $|P| \leq d n$. 
We need to compute lower bounds on the match rates.
For this, we need the upper bound on the total population, from Theorem~\ref{thm::total_size_upper_bound},
of $(3/2)nN+n\le (15/8)n\sqrt{T} +n\le (751/400)n\sqrt{T}$, as by
Constraint~\ref{eqn::constraints}, $n\ge T\ge 400$. Therefore, the number of men and the number of women in the system is at most $\frac{1}{2} ((751/400)n\sqrt{T} + (288 / T) (n \sqrt{T}) + n N / (25 \sqrt{T})) \le 13\sqrt{T}/10$.
The second term is from the last Type $2$ strip. $T \geq 400$ and $N \leq \frac{5}{4} \sqrt{T}$ are used in the last inequality.
}

Let $P_{t'',s}$ and $P_{t'',s'}$ be the portions of population $P_{t''}$ at time $t''$ in strips $s$ and $s'$, resp., for $t+1-\sqrt{T} \le t''\le t$.

By Lemma~\ref{lem::upper_strip_tech}, the matches remove at least the following number of people from $P_{t'',s}$:
\begin{align*}
    \frac{(|P_{t'',s}|^2)/2  -  (n / 25 \sqrt{T})^2 / 2}  n \sqrt{T} = \frac{ |P_{t'',s}|^2  -  n^2 / 625 T}{2  n \sqrt{T}} 
\end{align*}

A similar bound applies to the matches involving
$P_{t'',s'}$.
To minimize the terms $ |P_{t'',s'}|^2/2n\sqrt{T}$ for $P_{t'',s}$ and $P_{t'',s'}$,
we should make them equal.
Thus the expected number of matches of population $P_{t''}$ is at least

\begin{align}
\label{eqn::type1-matches-min}
\frac{|P_{t''}|^2} {4n\sqrt{T}} - \frac{ n}{625 T \sqrt{T}}.
\end{align}

Next, we want to obtain a high probability bound.

There are four sets of people, resp.\ the men and women in each of $P_{t'',s}$ and $P_{t'',s'}$. Let $\mu$ be the number of matches of one set. If $E[\mu] \geq \frac{\alpha n}{\sqrt{T}}$, then
\begin{align*}
    \Pr\big[\mu \ge  E[\mu](1-\epsilon) \big] \ge  1 - e^{-(E[\mu] \epsilon^2)/2} \ge 1 - e^{-\epsilon^2  \alpha n / (2 \sqrt{T})};
\end{align*}
letting $\epsilon  = \sqrt{\frac{2 \sqrt{T}}{\alpha n} \ln (4 T n^{2c + 1})}$ yields $\Pr\big[\mu \ge  E[\mu](1-\epsilon) \big]  \ge 1 - \frac{1}{4 T n^{2c + 1}}$.

Otherwise, $E[\mu] \leq \frac{\alpha n}{\sqrt{T}}$, and
\begin{align*}
    \Pr\big[\mu \ge  E[\mu] - \theta \big] \ge 1 - e^{-(\theta^2)/(2 E[\mu])} \ge 1 - e^{-(\sqrt{T} \theta^2) /(2 \alpha n)}; 
\end{align*}
setting $\theta = \sqrt{\frac{2\alpha n}{\sqrt{T}} \ln (4 \sqrt{T} n^{2c + 1})} = \alpha \epsilon \frac{n}{\sqrt{T}}$ yields  $\Pr\Big[\mu \ge  E[\mu] - \theta \Big] \ge 1 - \frac{1}{4 T n^{2c + 1}}$.

Recall \eqref{eqn::type-one-new-agent-bound}, the high probability bound that the number of new agents is at most $1.025 \frac{n}{\sqrt{T}}$.

By \eqref{eqn::type1-matches-min}, the number of people matched is at least $(1 - \epsilon)\left[\frac{|P_{t''}|^2} {4n\sqrt{T}} - \frac{ n}{625 T \sqrt{T}}\right] - 4 \theta$. 
Recall that $|P_{t''}| \leq d n$; we let $x = d n - |P_{t''}|$. Then, the number of people left is at most:
\begin{align*}
    d n - x - \left[(1 - \epsilon)\left[\frac{(d n - x)^2} {3n\sqrt{T}} - \frac{ n}{625 T \sqrt{T}}\right] - 4 \theta \right] \leq dn - \left[(1 - \epsilon)\left[\frac{(d n)^2} {4n\sqrt{T}} - \frac{ n}{625 T \sqrt{T}}\right] - 4 \theta \right],
\end{align*}
if $1 \geq (1 - \epsilon) d / 2\sqrt{T}$. This number is upper bounded by
$d n - (1 - \epsilon) (\frac{ d^2}{4} - \frac{1}{422,500}) \frac{n}{\sqrt{T}} +   4 \alpha \epsilon \frac{n}{\sqrt{T}}$ as $T \geq 676$ and $|P_{t''}| \geq dn$. 
Let $d = 2.6$ and $\alpha = \frac{3}{16}$. Also, $\epsilon \leq \frac{1}{10}$ as $n \ge 27 (2c + 2) T\ln n$ and $676 \le T \leq n$. 
A final calculation shows that $(1 - \epsilon) (\frac{ d^2}{4} - \frac{1}{422,500}) \frac{n}{\sqrt{T}} -   4 \alpha \epsilon \frac{n}{\sqrt{T}}$ is at least $1.025 \frac{n}{\sqrt{T}}$, demonstrating the result.

\end{proof}

\subsection{Upper Bound on the Size of a Type 2 Strip}
\label{appn::type2_upper}

\begin{proof} (of Theorem~\ref{type2_upper})

Consider any Type $2$ strip s. If $s$ is the topmost Type $2$ strip, clearly we can upper bound its size by twice the bound on the size of a Type $1$ strip given in Theorem~\ref{type1_upper}. In addition, if $s$ is the Type $2$ strip next to the topmost Type $2$ strip, then the size of strip $s$ is less than that of the topmost Type $2$ strip at time $t + 1 - \sqrt{T}$, which completes the proof for this strip too.

We now assume that $s$ has at least two Type $2$ strips above it. Let $s'$ be the strip immediately above $s$, and let $h$ denote the height of $s$. Then the height of $s'$ is $h/2$ and $h \geq \sqrt{T}$.  

Let's consider the agents who will be in strip $s$ at time $t+1$. They will all enter $s$ during a length $h$ time interval ending at time $t+1$. They can be partitioned into two sets as follows:
\begin{itemize}
    \item $Y_{t+1} = \{\text{agents that will have spent less than $h/2$ time in strip $s$ by time $t+1$}\}$. .
    \item $O_{t+1} =\{\text{agents that will have spent at least $h/2$ time steps in strip $s$ by time $t+1$}\}$.
\end{itemize}

The agents in $Y_{t+1}$ were all present at time $t'=t+1-h/2$ as part of the population of strip $s'$ at that time. 
By the inductive hypothesis, applied to $s'$ at time $t'$, we know that there were at
most $2 \cg n\sqrt{T}/h$ agents in $s'$ at that time.
Let $P_y$ denote this population.
The agents in $O_{t+1}$ were all present at time $t'=t+1-h$ as part of the population of strip $s'$ at that time. 
By the inductive hypothesis, applied to $s'$ at time $t'$, we know that there were at
most $2\cg n\sqrt{T}/h$ agents in $s'$ at that time.
Let $P_o$ denote this population.

Let $P^{y}_{t'',s}$ and $P^y_{t'',s'}$ be the remainder of population $P_y$ at time $t''$ in strips $s$ and $s'$, resp., for $t+1-\sqrt{T} \le t''\le t$.
Also, let $P^y_{t''} = P_{t'',s} \cup P_{t'',s'}$. Similarly, define $P^{o}_{t'',s}$, $P^o_{t'',s'}$ and $P^o_{t''}$.

\hide{We will now show that both populations $P_y$ and $P_o$ must have shrunk substantially over $h/2$ and $h$ steps respectively.}

To this end, we need to compute lower bounds on the match rates.

\rjc{By \eqref{eqn::upper_and_lower_TotBound},
the maximum of the number of men and number of women in the market is at most $n\sqrt{T}$.
}
\hide{
We recall the upper bound on the total population, from Theorem~\ref{thm::total_size_upper_bound},
of $(3/2)nN+n\le (15/8)n\sqrt{T} +n\le  (751/400)n\sqrt{T}$, as by
Constraint~\ref{eqn::constraints}, $n\ge T\ge 400$. Therefore, the number of men and number of women in the market should not exceed $\frac{1}{2} ((751/400)n\sqrt{T} + (288 / T) (n \sqrt{T}) + n N / (25 \sqrt{T})) \leq 13 n\sqrt{T}/10 $. The second term is from the last type $2$ strip. $T \geq 400$ and $N \leq \frac{5}{4} \sqrt{T}$ is used in the last inequality.
}

We divide the period $[t+1-h,t+1)$ into two phases; Phase 1, $[t+1-h,t+1-h/2)$, and Phase $2$, $[t+1-h/2,t+1)$. We will show that the size of $P^y_{t''}$ at the end of Phase $1$ is at most $\cg_1n\sqrt{T}/h$. We will specify $\cg_1$ later. Then, at the start of Phase $2$ the size of $P^o\cup P^y$ is at most $2\cg n\sqrt{T}/h + \cg_1n\sqrt{T}/h$. We claim that after Phase $2$, the size of $P^o\cup P^y$ is at most $\cg n\sqrt{T}/h$.

We analyze Phase $1$ first.
Consider the set $P^o_{t'',s}$ and time $t'' \in [t + 1 - h, t +1 - h /2)$. 
 
By Lemma~\ref{lem::upper_strip_tech}, at time $t''$ these matches remove, in expectation, at least the following number of people from $P^o_{t'',s}$:
\begin{equation}
\label{eqn::type2_matches}
    \begin{aligned}
    \frac{ |P^o_{t'',s}|^2  -  n^2 / \rjc{625} T}{2  n \sqrt{T}} 
    \end{aligned}
\end{equation}

Similar bounds will hold for the sets $P^o_{t'',s'}$. 
Notice that the total expected number of matches from $P^o_{t''}$ is minimized if $|P^o_{t'',s}|=|P^o_{t'',s'}|$. Thus we obtain that the size of $P^o_{t''}$ reduces, in expectation, by at least 

\begin{equation}
\label{eqn::type2_old_reduction_late_phase}
\begin{aligned}
\frac{|P^o_{t''}|^2} {4 n\sqrt{T}}- \frac{n}{625 T\sqrt{T}}.
\end{aligned}
\end{equation}

As in the analysis for the Type $1$ strip, we then give a high probability bound.
We have four sets of people, the men and the women in the sets $P^o_{t'',s'}$ and $P^o_{t'',s}$, respectively. Suppose $\mu$ be the number of matches in one of these set at time $t''$.

If $E[\mu]\geq \alpha n\sqrt{T}/h^2$, by Lemma \ref{lem::negative_dependence_two_sex},

\begin{equation*}
    \begin{aligned}
    \Pr\Big[\mu \ge  E[\mu](1-\epsilon) \Big] \ge  1 - e^{-(E[\mu] \epsilon^2)/2} \ge 1 - e^{-\epsilon^2 \alpha n \sqrt{T}/2h^2},
    \end{aligned}
\end{equation*}

Setting $\epsilon=\bigg[\frac{2h^2}{\alpha n\sqrt{T}}\ln(T (\log_2 \sqrt{T} + 1) n^{2c+1})\bigg]^\frac{1}{2}$ yields  $\Pr\Big[\mu \ge  E[\mu](1-\epsilon) \Big] \ge  1 - \frac{1}{T (\log_2 \sqrt{T} + 1) n^{2c+1}}$.

Otherwise, $E[\mu]\leq \alpha n\sqrt{T}/h^2$, so by Lemma \ref{lem::negative_dependence_two_sex},

\begin{equation*}
    \begin{aligned}
    \Pr\Big[\mu \ge  E[\mu] - \theta \Big] \ge 1 - e^{-(\theta^2)/(2 E[\mu])} \ge 1 - e^{-(\theta^2 h^2) /(2 \alpha n\sqrt{T})},  
    \end{aligned}
\end{equation*}

Setting $\theta=\bigg[\frac{2\alpha n\sqrt{T}}{h^2}\ln(T (\log_2 \sqrt{T} + 1)  n^{2c+1})\bigg]^\frac{1}{2} = \epsilon \alpha \frac{n \sqrt{T}}{h^2}$ yields  $\Pr\Big[\mu \ge  E[\mu]-\theta) \Big] \ge  1 - \frac{1}{T (\log_2 \sqrt{T} + 1) n^{2c+1}}$.

For each  of the four sets we can use one of the two bounds above.

We can set $\alpha=0.1,\epsilon\leq 0.1$ by imposing the constraints $c\geq 1, 400\leq T\leq n, n\geq 125(2c+2.5)T\sqrt{T}\ln n$ which are provided by the constraints in \eqref{eqn::constraints}, $h \leq T / 4$, and $\log_2 \sqrt{T} + 1\leq \sqrt{T} / 4$ .

Let $\cg(\cdot)$  be a real valued function. 
Suppose the size of the set $P^o_{t''}$ at round $t''$ is smaller than $\cg(t'') \cdot  n \sqrt{T} / h$ and let $X = \cg(t'') \cdot n \sqrt{T} / h - |P^o_{t''}|$. If $(1 - \epsilon)  \cg(t'') / (2 h) \leq 1$, \footnote{Note that this is satisfied when $\cg(t'') \leq 21.5$ and $h \geq \sqrt{T} = 20$.} then the size of the set $P^o_{t'' + 1}$ at round $t'' + 1$ is at most
\begin{align*}
    &\cg(t'') n \sqrt{T} / h - X - (1 - \epsilon) \left[\frac{(\cg(t'') n \sqrt{T} / h - X)^2} {4 n\sqrt{T}}- \frac{ n}{625 T\sqrt{T}}\right] + 4  \epsilon \alpha \frac{n \sqrt{T}}{h^2} \\
    &\vspace*{0.4in} \leq \cg(t'') n \sqrt{T} / h  - (1 - \epsilon) \left[\frac{(\cg(t'') n \sqrt{T} / h )^2} {4 n\sqrt{T}}- \frac{ n}{625 T\sqrt{T}}\right] + 4  \epsilon \alpha \frac{n \sqrt{T}}{h^2} \\
    &\vspace*{0.4in} \leq \cg(t'') n \sqrt{T} / h -  \frac{9\cg(t'')^2 n \sqrt{T} } {40 h^2 } + \frac{9 n}{6250 T\sqrt{T}} +  \frac{n \sqrt{T}}{25 h^2} \\
     &\vspace*{0.4in} \leq  n \sqrt{T} / h  \left[ \cg(t'') - \frac{1}{h} \left(\frac{9\cg(t'')^2 } {40} - 0.041\right)\right].
\end{align*}
The last inequality uses the constraint that $h \leq T / 4$.

Let $\cg(t + 1 - h) = 2\cg$ and $\cg(t''+ 1) = \left[ \cg(t'') - \frac{1}{h} \left(\frac{9\cg(t'')^2 } {40} - 0.041\right)\right]$ for $t'' \in [t + 1 - h, t + 1 - h/2)$, then we have shown that the size of the set $P^o_{t''}$ at round $t + 1 - h/2$ is at most $\cg(t+1 - h/2) \cdot n \sqrt{T} / h$. One way to solve $\cg(\cdot)$ by using a differential equation.
Consider the differential equation $\dd \bar{\cg} / \dd t = - \frac{1}{h} \left(\frac{9\bar{\cg}^2 } {40} - 0.041\right)$ and $\bar{\cg}(t + 1 - h) = 2\cg$. Note that $\bar{\cg}(t'') \geq \cg(t'')$ for all $t'' \in [t + 1 - h, t+ 1 - h/2]$. \footnote{Suppose it is not true. Since $\bar{\cg}(t + 1  - h) =  \cg(t + 1 - h) = 2\cg$, there exists a $t'$, such that $\bar{\cg}(t') \geq  \cg(t')$ and $\bar{\cg}(t' + 1) < \cg(t'+1)$. Then, there exist a $t'' \in [t', t' + 1)$ such that $\bar{\cg}(t'') = \cg(t')$. After time $t''$, $\dd \bar{\cg}(t'') / \dd t \geq - \frac{1}{h} \left(\frac{9\bar{\cg(t'')}^2 } {40} - 0.041\right) = \cg(t' + 1) - \cg(t')$. Therefore, $\bar{\cg}(t'+1) = \bar{\cg}(t'') + \int_{s = t''}^{t'+1} \dd \bar{\cg}(s)  \geq \cg(t') + \int_{s = t''}^{t'+1} [\cg(t' + 1) - \cg(t')] \dd s \geq \cg(t' + 1)$, which contradicts the assumption.}

Therefore, in order to prove $\cg(t+1 - h/2) \leq \cg_1$, we only need $\bar{\cg}(t+1 - h/2) \leq \cg_1$. We look at the total time for $\bar{\cg}$ to reduce from the value $\cg(t+1 - h) = 2\cg$ to $\cg_1$: $\dd t = - h  \dd \bar{\cg} / (\frac{9 \bar{\cg}^2}{40} - 0.041) $. Therefore, the total time is $\int_{\bar{\cg} = \cg_1}^{2\cg} h  \dd \bar{\cg} / (\frac{9 \bar{\cg}^2}{40} - 0.041) <= \int_{\bar{\cg} = \cg_1}^{2\cg} h  \dd \bar{\cg} / (\frac{9 \bar{\cg}^2}{40} - 0.041 \frac{\bar{\cg}^2}{\cg_1^2}) = (h / (\frac{9 }{40} - \frac{0.041}{\cg_1^2}))(1/\cg_1 - 1 / (2\cg))$. To have this be at most $h/2$ (the total duration of Phase $1$), we only need $2(1/\cg_1 - 1 / (2\cg)) \leq 9/40 - 0.041/(\cg_1)^2$, which is satisfied by letting $\cg = 7.5$ and $\cg_1 = 6.5$.

We consider Phase $2$ next. Consider the set $P^o_{t'',s} \cup P^y_{t'',s}$ and time $t'' \in [t + 1 - h/2, t +1)$. The analysis is exactly the same as that for Phase $1$.
By Lemma~\ref{lem::upper_strip_tech}, these matches remove, in expectation, at least the following number of people from $P^o_{t'',s} \cup P^y_{t'',s}$:
\begin{equation}
    \begin{aligned}
    \frac{ |P^o_{t'',s} \cup P^y_{t'',s}|^2  -  n^2 / 625 T}{2  n \sqrt{T}} 
    \end{aligned}
\end{equation}
Similar bounds will hold for the set $P^o_{t'',s'} \cup P^y_{t'',s'}$. We also reduce the size of $P^o_{t''} \cup P^y_{t''}$, in expectation, by at least 
\begin{equation}
\begin{aligned}
\frac{|P^o_{t''} \cup P^y_{t''}|^2} {3 n\sqrt{T}}- \frac{ n}{625 T\sqrt{T}}.
\end{aligned}
\end{equation}

Then, as in Phase $1$, suppose the size of the set $P^o_{t''} \cup P^y_{t''}$ at time $t''$ is smaller than $g(t'') n \sqrt{T} / h$ and let $X = g(t'') n \sqrt{T} / h - |P^o_{t''} \cup P^y_{t''}|$. If $(1 - \epsilon)  g(t'') / (2 h) \leq 1$, then the size of the set $P^o_{t''} \cup P^y_{t''}$ at round $t'' + 1$ is at most
\begin{align*}
      n \sqrt{T} / h  \left[ g(t'') - \frac{1}{h} \left(\frac{9g(t'')^2 } {40} - 0.041\right)\right].
\end{align*}

Let $\cg(t + 1 - h/2) = 2\cg + \cg_1$ and $g(t''+ 1) = \left[ g(t'') - \frac{1}{h} \left(\frac{9\cg(t'')^2 } {40} - 0.041\right)\right]$ for $t'' \in [t + 1 - h/2, t + 1)$, then we have shown that the size of the set $P^o_{t''} \cup P^y_{t''}$ at round $t + 1$ is at most $\cg(t + 1) \cdot n \sqrt{T} / h$. We consider the same differential equation here, $\dd \bar{\cg} / \dd t = - \frac{1}{h} \left(\frac{9\bar{\cg}^2 } {40} - 0.041\right)$, with $\bar{\cg}(t + 1 - h/2) = 2\cg + \cg_1$. Note that $\bar{\cg}(t'') \geq \cg(t'')$ for all $t'' \in [t + 1 - h/2, t+ 1]$. 

Therefore, in order to prove $\cg(t+1) \leq \cg$, we only need $\bar{\cg}(t+1) \leq \cg$. We look at the total time for $\bar{\cg}$ to reduce from the value $2g + g_1$ to $\cg$: $\dd t = - h  \dd \bar{\cg} / (\frac{9 \bar{\cg}^2}{40} - 0.041) $. Therefore, the total time is $\int_{\bar{\cg} = \cg}^{2\cg + \cg_1} h  \cdot \dd \bar{\cg} / (\frac{9 \bar{\cg}^2}{40} - 0.041) <= \int_{\bar{\cg} = \cg}^{2\cg + \cg_1} h \cdot  \dd \bar{\cg} / (\frac{9 \bar{\cg}^2}{40} - 0.041 \frac{\bar{\cg}^2}{\cg^2}) = (h / (\frac{9 }{40} - \frac{0.041}{\cg^2}))(1/\cg - 1 / (2\cg + \cg_1))$. To have this be at most $h/2$ (the total duration of Phase $2$), we only need $2(1/\cg - 1 / (2\cg + \cg_1)) \leq 9/40 - 0.041/(\cg)^2$, which is also satisfied by letting $\cg = 7.5$ and $\cg_1 = 6.5$.

Finally, as there are $\log_2 \sqrt{T} + 1$ Type $2$ strips, and, for each Type $2$ strip, we consider $h \leq T/4$ steps, the success probability is at least $1 - \frac{T/4 \cdot 4 \cdot (\log_2 \sqrt{T} + 1)}{T(\log_2 \sqrt{T} + 1) n^{2c+1}} \leq 1 - \frac{1}{n^{2c + 1}}$.

\end{proof}

\begin{proof} (Of Theorem~\ref{thm::last_strip_size_upper_bound}.)~
Let's call the bottommost Type $2$ strip $s$ and the Type $2$ strip immediately above it strip $s'$.
Any agent in the population in $s$ at time $t+1$ must belong to one of the following categories:

\begin{itemize}
    \item The agent was present in $s'$ at time $t+1-T/4$.
    \item The agent was present in $s'$ at time $t+1-T/2$.
\end{itemize}

However, by our inductive hypothesis $H(t)$, we know that at all time steps before $t+1$, the size of $s$ was always less than $\frac{7.5n\sqrt{T}}{\text{height of strip $s$}}\leq 4\cdot 7.5 n/\sqrt{T}.$ 

Thus,  the size of $s$ at time $t+1$ is bounded by $8\cdot 7.5n/\sqrt{T}\leq 60n/\sqrt{T}$, which concludes the proof.
\end{proof}

\subsection{Bound on the Imbalance}
\label{appn::imbalance}

\begin{proof} (Of Claim~\ref{clm::update-to-I}).
The expected number of matches at time $\tau$ between men in diagonal $d_i$ and women in diagonal $d_j$ is 
$$\frac{(2A_i+I_i+X_i)(2A_j-I_j-X_j)}{4R}.$$
Similarly, the expected number of women in $d_i$ that match with men in $d_j$ is 
$$\frac{(2A_i-I_i-X_i)(2A_j+I_j+X_j)}{4R}.$$
Thus $I'(d_i,\tau)$ is given by:
\begin{align*}
 & I_i+X_i - \sum_{d_j \in s}\Big[ \frac{(2A_i+I_i+X_i)(2A_j-I_j-X_j)}{4R}-\frac{(2A_i-I_i-X_i)(2A_j+I_j+X_j)}{4R}\Big]\\
&~~= I_i+X_i - \sum_{d_j \in s} \Big[ X_i\frac{(2A_j-I_j-X_j)}{4R}-X_j\frac{(2A_i-I_i-X_i)}{4R}+I_i\frac{(2A_j-I_j-X_j)}{4R}-I_j\frac{(2A_i-I_i-X_i)}{4R}\Big].
\end{align*}

\end{proof}

\begin{proof} (Of Claim~\ref{clm::bound-on-Y-one-strip}.)~
We first give a high probability bound on $\sum Y(d_i, \tau)$.
Let $m(d_i, \tau)$ be the number of men entering the market on diagonal $d_i$ at time $\tau$. Note that $d_T$ is the last Type $1$ diagonal. Let $d_r$ be a diagonal in Type $1$ strip $s$; then,
\begin{align*}
\chance\bigg[\Big|\sum_{r \leq i \leq T} Y(d_i, \tau)\Big| > \Delta \bigg] = \chance\bigg[\Big|\sum_{r \leq i \leq T} m(d_i, \tau)  - (T-r+1)\cdot \frac{n}{2T}\Big| >  \Delta / 2\bigg].
\end{align*}
Note that 
\begin{align*}
    &\chance\bigg[\Big|\sum_{r \leq i \leq T} m(d_i, \tau)  - (T-r+1)\cdot \frac{n}{2T}\Big| >  \Delta / 2\bigg]\\
    &\leq 2 \exp\Big[-\Delta^2  / \Big(3 \cdot (T-r+1)\cdot \frac{n}{2T}\Big)\Big] \leq 2 \exp \Big[-\Delta^2 / \Big(\frac{3n}{2}\Big)\Big].
\end{align*}
The last inequality follows as $T - r + 1 \leq T$. Letting $\Delta = \sqrt{\frac{3n}{2} \ln \left(2 T n^{3c + 1}\right)}$ yields
\begin{align*}
\chance\bigg[\Big|\sum_{r \leq i \leq T} Y(d_i, \tau)\Big| > \sqrt{\frac{3n}{2} \ln \left(2 T n^{3c + 1}\right)} \bigg] \leq \frac{1}{T n^{3c + 1}}.
\end{align*}
Therefore, with probability at least $\frac{1}{n^{3c+1}}$, for all $r$ such that $d_r$ is a diagonal in a Type $1$ strip, $\left|\sum_{r \leq i \leq T} Y(d_i, \tau)\right| \leq \sqrt{\frac{3n}{2} \ln \left(2 T n^{3c + 1}\right)}$.

With this result in hand, we prove the claim as follows.
Let $d_{r(s)}$ be the rightmost (lowest index) diagonal in $s$ and $d_{l(s)}$ be the leftmost (highest index) diagonal in $s$. Let $w_i = \sum_{j\ge r(s)} Y(d_i,\tau,d_j,\tau')/Y(d_i,\tau)$.
Let's consider 
$\sum_{d_i\in s'; j \ge r(s)} Y(d_i,\tau,d_j,\tau')
= \sum_{d_i\in s'} w_i\cdot Y(d_i,\tau)
$. 
By Claim~\ref{clm::distr-of-X}, $w_i \le w_k$, for $i<k$. Thus, 
\begin{align*}
    \Big|\sum_{d_i; j \ge r(s)} Y(d_i,\tau,d_j,\tau')\Big|=
    \Big|\sum_{i = 1}^{T} w_i\cdot Y(d_i,\tau)\Big|\leq w_T \max_{r \leq T}\big| \sum_{r\le i \le T} Y(d_i,\tau)\Big| \leq \sqrt{\frac{3n}{2} \ln \left(2 T n^{3c + 1}\right)}.
\end{align*}
Finally,
\begin{align*}
   \Big|\sum_{d_i; d_j\in s} Y(d_i,\tau,d_j,\tau') \Big| =
   \Big| \sum_{d_i; j \ge r(s)} Y(d_i,\tau,d_j,\tau') -
   \sum_{d_i; j \ge l(s)+1} Y(d_i,\tau,d_j,\tau') \Big| \le 2\sqrt{\frac{3n}{2} \ln \left(2 T n^{3c + 1}\right)},
\end{align*}
which completes the proof.
\end{proof}

\begin{proof} (Of Claim~\ref{clm::remain::type::2}.)~
We begin by bounding $\sum_{j: d_j \in s} (2A_j - I_j - X_j) / 4R$ for any Type $2$ strip $s$. 
\begin{align}
\label{eqn::match-rate-bound-second}
\sum_{j: d_j \in s} \frac{(2A_j - I_j - X_j)}{4R} \leq \frac{1}{2} \frac{\frac{3.75n\sqrt{T}}{H}+n/50\sqrt{T}}{\frac{n\sqrt{T}}{6}}<\frac{23}{2H}\hspace*{0.2in}\text{(as $\sqrt{T}\ge 26$ by constraint~\ref{eqn::constraints})}.
\end{align}

Let $s$ have height $H$.

Consider $X(d,\tau,d',\tau')$. If $d'$ is in a Type $2$ strip then by \eqref{eqn::match-rate-bound-second} at most $\frac{23}{2H}$ of it disperses to some location in the same strip and at least $1-\frac{23}{2H}$ of it moves down distance one. This implies that, within $H$ time, a Type $2$ strip loses at least $e^{-23/2}$ of the $X(d,\tau,d',\tau')$ that had been present within it at time $\tau'$. Let $K_2=e^{12}\ln 2$. Therefore, by time $\tau'+ K_2 H$ ($\leq \tau' K_2 T / 4$) at least half of the $X(d,\tau,d',\tau')$ in a Type $2$ strip has moved out of the strip.

Similarly, we can now carry out the same kind of argument for the Type $2$ strips. After $2e^2(\ln 2) \sqrt{T}(\sqrt{T}+\log_{2}(2n^k))$ time there is at most $\frac{1}{2n^k}$ fraction of $X(d,\tau)$ in the Type $1$ strips. We focus on the remaining
$1 - \frac{1}{2n^k}$ portion of $X(d,\tau)$ which has already entering the type 2 strips.
Number the Type $2$ strips from top to bottom\footnote{Our argument doesn't involve the last Type $2$ strip, so we will end at the second to last strip.}. Now consider $\gamma$ as a distribution of the rest of the $X(d,\tau)$ where $\gamma_i$ is the fraction of $X(d,\tau)$ in strip number $i$. Recall that there are 
$\log_2(\sqrt{T})$ Type $2$ strips other than the bottom Type $2$ strip. We consider the worst case where  all the remaining $(1 - \frac{1}{2 n^k}) \cdot X(d,\tau)$ starts out at the topmost Type $2$ strip. Define a potential function $\phi(\gamma)=\sum_{i=1}^{\log_2 \sqrt{T} + 1}\gamma_i\cdot 2^{(\log_2 \sqrt{T}) - i + 1}$. For the remaining $(1 - \frac{1}{2 n^k}) \cdot X(d,\tau)$, The initial potential is at most $\sqrt{T}$. Every $K_2 T / 4$ time steps, the potential decreases by at least $1/4$. Therefore, after $\frac{1}{\log_2 (4/3)} K_2T/4\log_{2}(2n^k\sqrt{T})$ time steps, the potential would have reduced to at most $\frac{1}{2n^k}$. 

There is also $\frac{1}{2n^k}$ fraction which might still be in the Type $1$ strips. Thus the fraction of $X(d,\tau)$ in any strip other than the bottommost Type 2 strip after $\frac{1}{\log_2 (4/3)} K_1\sqrt{T}(\sqrt{T}+\log_{2}(2n^k))+$ 
\newline$\frac{1}{\log_2 (4/3)} K_2T/4\log_{2}(2n^k\sqrt{T})$ time is at most $\frac{1}{n^k}$.
\end{proof}

\begin{proof} (Final Calculation in Theorem ~\ref{thm::imbalance_bound})\\

As $\kappa = \frac{e^2 \ln 2}{\log_2 (4/3)} \sqrt{T}(\sqrt{T}+\log_{2}(2n^k)) + \frac{e^{12} \ln 2}{4 \log_2 (4/3)} T\log_{2}(2n^k\sqrt{T}) \leq 12.35 (T + \sqrt{T} + (c+4) \sqrt{T}\log_2 n ) + \frac{e^{12}}{2} (T + (c + 4) T \log_2 n + 0.5 T \log_2 T)$, the total imbalance is at most 
\begin{align*}
    \frac{15T}{n^3} + &\Big[12.35 (T + \sqrt{T} + (c+4) \sqrt{T}\log_2 n ) + \frac{e^{12}}{2} (T + (c + 4) T \log_2 n + 0.5 T \log_2 T) \Big] \cdot \\
    &~~~~~~~~~~~~~\bigg(192\sqrt{ \frac{n\ln(4n^{3c+1} (T^2/32 + T/8) N)}{\sqrt{T}}} +2\sqrt{\frac{3n}{2} \ln \left(2 T n^{3c + 1}\right)} \bigg).
\end{align*} In order to make it smaller than $\frac{n}{25 \sqrt{T}}$, we only need that

\begin{align*}\frac{375 T \sqrt{T}}{n^3} + &\Big[309 (T + \sqrt{T} + (c+4) \sqrt{T}\log_2 n ) + \frac{25 e^{12}}{2} (T + (c + 4) T \log_2 n + 0.5 T \log_2 T)\Big] \cdot \\
&~~~~~~~~~~~~\bigg(192\sqrt{ \sqrt{T}\ln[4n^{3c+1} (T^2/32 + T/8) N]} +2\sqrt{\frac{3 T}{2} \ln \left(2 T n^{3c + 1}\right)} \bigg) \leq \sqrt{n}.
\end{align*}

As $375 T\sqrt{T} / n^3 \leq 0.0012 \sqrt{n}$ from the constraint $n \geq T \geq 676$, we need
\begin{align*} &\Big[309 (T + \sqrt{T} + (c+4) \sqrt{T}\log_2 n ) + \frac{25 e^{12}}{2} (T + (c + 4) T \log_2 n + 0.5 T \log_2 T)\Big] \cdot \\
&~~~~~~~~~~~~\bigg(192\sqrt{ \sqrt{T}\ln[4n^{3c+1} (T^2/32 + T/8) N]} +2\sqrt{\frac{3 T}{2} \ln \left(2 T n^{3c + 1}\right)} \bigg) \leq 0.9988\sqrt{n}.  \numberthis \label{ineq::final::1}
\end{align*}
In addition, as $n \geq T \geq 676$,
\begin{align*}
    &\Big[309 (T + \sqrt{T} + (c+4) \sqrt{T}\log_2 n ) + \frac{25 e^{12}}{2} (T + (c + 4) T \log_2 n + 0.5 T \log_2 T)\Big] \\
    &~~~~~~~~~~~~~~~~~~\leq (86.61 + 12.876c + 57.62 e^{12} + 12.5 e^{12} c) T \log_2 n, \numberthis \label{ineq::final::2}
\end{align*}
and, as $n \geq T \geq 676$ and $n \geq N$,
\begin{align*}
    \bigg(192\sqrt{ \sqrt{T}\ln[4n^{3c+1} (T^2/32 + T/8) N]} +2\sqrt{\frac{3 T}{2} \ln \left(2 T n^{3c + 1}\right)} \bigg) \leq 42 \sqrt{T (3c + 4) \ln n}.  \numberthis \label{ineq::final::3}
\end{align*}
By inequalities \eqref{ineq::final::1}, \eqref{ineq::final::2}, and \eqref{ineq::final::3},
\begin{align*}
    (86.61 + 12.876c + 57.62 e^{12} + 12.5 e^{12} c) T \log_2 n \cdot 42 \sqrt{T (3c + 4) \ln n} \leq 0.9988\sqrt{n}.
\end{align*}
Therefore, $n \geq (3654 + 2436e^{12} + 546(e^{12} + 1) c)^2(3c+4) T^3 (\log_2 n)^2 \ln n$ suffices.

\end{proof}

\subsection{Initialization}
\label{appn::init}

\begin{proof} (Of Theorem~\ref{thm::initialization}.)~
At any point in the first $\sqrt{T}$ time steps:
\begin{itemize}
    \item The total population in the entire matching pool is clearly less than $n\sqrt{T}<nN$, as only these many agents could have even entered the matching pool.
    
    \item In any single Type $1$ strip, 

$$\pr{\text{Number of agents that entered the strip from the top}\leq \frac{n\sqrt{T}(1+\epsilon)}{\sqrt{T}}}\geq 1-e^{-\frac{n\epsilon^2}{3}}.$$

However the agents that enter a Type $1$ strip during the first $\sqrt{T}$ time must either have entered from the top or they could have entered at the top boundary of the previous strip. Thus, by a union bound,

$$\pr{\text{Number of agents that entered any Type $1$ strip}\leq 2n(1+\epsilon)}\geq 1-\sqrt{T}e^{-\frac{n\epsilon^2}{3}}.$$

So by setting $\epsilon=\sqrt{\frac{3}{n}\ln (n^{c+1}\sqrt{T})}$, and imposing the constraints $c\ge 1$, $T\le n$, and $n\ge 35(c+2)\ln n$ that guarantee that $\epsilon< 0.3$ (from \eqref{eqn::constraints}), we obtain that with probability $1-\frac{1}{n^{c+1}}$ every Type $1$ strip has a population $<2.6n$.

\item The agents in the first Type $2$ strip after $\sqrt{T}$ time steps (the only Type $2$ strip with any population after $\sqrt{T}$ time) could only be those agents that entered the leftmost two Type $1$ strips from the top. However the previous bound already guarantees that this number is also $<2.6n$.

\item Also, the population in the bottom most Type 2 strip will be 0.

\item Now it remains only to show that in each of the strips, except possibly the bottommost Type 2 strip, $|\text{number of men}-\text{number of women}|\leq \frac{n}{25\sqrt{T}}$. 
\end{itemize}

We will follow the proof of Theorem \ref{thm::imbalance_bound}, though the proof will be simplified by the fact that we only need to consider $\sqrt{T}$ many time steps.

We divide each strip into thin diagonals of width $1$. Let the diagonal include the bottom but not the top boundary. Notice that for each value, a diagonal contains at most one grid point. 

As in Theorem \ref{thm::imbalance_bound}, we introduce the following notation w.r.t.\ diagonal $d$ at time step $\tau$, where we are conditioning on the outcome of step $\tau-1$.
\begin{align*}
I(d,\tau) &= \Expect[(\text{number of men at time $\tau$}-\text{number of women at time $\tau$})]\\
X(d,\tau) &=  (\text{number of men matching at time $\tau$}-\text{number of women matching at time $\tau$})\\
&\hspace*{0.2in} - \Expect[(\text{number of men matching at time $\tau$}-\text{number of women matching at time $\tau$})]\\
Y(d,\tau) &= \text{number of men entering at time $\tau$} - \text{number of women entering at time $\tau$}
\\
A(d,\tau) &= (\text{number of men matching at time $\tau$}+\text{number of women matching at time $\tau$})/2.
\end{align*}

$I(d,\tau)$ is measured after the entry of the new agents at time $\tau$ but prior to the match for this step. Also, note that $Y(d, \tau) = 0$ if $d$ is in a Type $2$ strip. 

In addition, observe that the imbalance $\Imb(s)$ at the start of step $t$ equals $\sum_{d\in s} I(d,t)$.

We observe that a match between two agents in distinct diagonals of the same strip
will increment the $(\text{number of men } - \text{ number of women})$
in one diagonal and decrement it in the other.
Thus there is a zero net change over all the diagonals
in the strip due to the matches. However, as the agents all age by 1 unit during a step, some agents enter the strip and some leave, which can cause changes to the imbalance within a strip.
However, the entry of new agents can introduce new imbalances.
We will need to understand more precisely how these imbalances evolve.

It is convenient to number the diagonals as $d_1,d_2,d_3,\ldots$, in right to left order.

We recall the following claims from the proof of Theorem \ref{thm::imbalance_bound}.

\begin{claim}
\label{clm::copy_update-to-I}
Let $d_i$ and $d_j$ be two diagonals in the same strip $s$. For brevity, let $I_i\triangleq I(d_i,\tau-1)$,
$I_j\triangleq I(d_j,\tau-1)$,
$A_i\triangleq A(d_i,\tau-1)$,
$A_j\triangleq A(d_j,\tau-1)$,
$X_i\triangleq X(d_i,\tau-1)$,
$X_j\triangleq X(d_j,\tau-1)$.
Finally, let $R$ denote the maximum of the total number
of men and the total number of women in the system
at time $\tau-1$.
Then the new imbalance on diagonal $d_i$, prior to every unmatched agent adding 1 to their age (which causes the agents on $d_i$ to move to $d_{i+1}$), denoted by $I'(d_i,\tau)$, is given by:
\begin{align*}
&I'(d_i,\tau)= \\
&~~~~I_i + X_i - \sum_{d_j \in s}\Big[   X_i\frac{(2A_j-I_j-X_j)}{4R}-X_j\frac{(2A_i-I_i-X_i)}{4R}+I_i\frac{(2A_j-I_j-X_j)}{4R}-I_j\frac{(2A_i-I_i-X_i)}{4R}\Big]; \\
&\text{~~and~~}I(d_i, \tau) = I'(d_{i-1}, \tau - 1) + Y(d,\tau).
\end{align*}
\end{claim}

$X(d,\tau)$ and $Y(d,\tau)$ are generated at diagonal $d$ at time $\tau$ and, by 
Claim \ref{clm::copy_update-to-I}, at any subsequent time step, $X(d,\tau)$ and $Y(d,\tau)$ will be redistributed over other diagonals.
\begin{enumerate}
    \item Due to the expected matching at time $\tau'\ge \tau$, each $X(d,\tau)$ and $Y(d,\tau)$ spreads to other diagonals in the same strip.
    \item At the end of  time step $\tau'$ the portions of $X(d,\tau)$ and $Y(d,\tau)$ present on diagonal $d_i$ move to diagonal $d_{i+1}$.
\end{enumerate}

Building on these observations, we will show our bound on the imbalance by bounding the total contribution from $X(\cdot, \tau)$ and $Y(\cdot, \tau)$ to strip $s$ at time $\tau'$.

Notice that  $\sum_{d_i\in s} I'(d_i,\tau) = \sum_{d_i\in s} I(d_i,\tau-1)$, for the coefficients
multiplying $X_i$ cancel, as they also do for $I_i$. Thus we can think of this process as redistributing the imbalance, but not changing the total imbalance.

Over time an imbalance $X(d_i,\tau)$ will be redistributed over many diagonals. We write
$X(d_i,\tau,d_j,\tau')$ to denote the portion of
$X(d_i,\tau)$ on diagonal $d_j$ at time $\tau'$.
$d_j$ need not be in the same strip as $d_i$.
Note that $\sum_{d_j} X(d_i,\tau,d_j,\tau') = X(d_i,\tau)$ for all $\tau'\ge \tau$. $Y(d_i,\tau,d_j,\tau')$ is defined analogously.

For the purposes of the following claim, we treat the final strip as a single diagonal, and in addition ignore the fact that people depart at age $T$ (which means that once an imbalance appears in this strip it remains there). The reason this strip is different is that it covers the whole of the bottom boundary and so is the only strip from which people leave the system by aging out.

\begin{claim}
\label{clm::copy_distr-of-X}
For all $\ell$, for all $i<k$, and for all $\tau'\ge \tau$,
$\big|\sum_{j>\ell} X(d_i,\tau,d_j,\tau')\big| \le \big|\sum_{j>\ell} X(d_k,\tau,d_j,\tau')\big|$.
The same property holds for the $Y(d_i,\tau,d_j,\tau')$.
\end{claim}

Later, we will show a common bound $B$ on the sums
$\big| \sum_{i\le j \le k} X(d_j,\tau)\big|$,
which holds for all $d_i$ and $d_k$ in the same strip and all $\tau$.

With this bound and Claim~\ref{clm::copy_distr-of-X} in hand, for each strip $s$, we can bound the contribution of the $X(d_i,\tau,d_j,\tau')$ summed
over all $d_i$ and over $d_j\in s$ by $2B$.
\begin{claim}
\label{clm::copy_bound-on-X-one-strip}
For all $\tau'\ge \tau$, for every strip $s$,
$\big|\sum_{d_i; d_j\in s} X(d_i,\tau,d_j,\tau')  \big| \le 2B$.
\end{claim}

This allows us to obtain the bound the imbalance in a strip $s$ at any time $\tau'\leq\sqrt{T}$ by considering the contributions of $\big|\sum_{d_i; d_j\in s} X(d_i,\tau,d_j,\tau')\big|$ and $\big|\sum_{d_i; d_j\in s} Y(d_i,\tau,d_j,\tau') \cdot Y(d_i, \tau) \big|$ at all possible previous times (which is at most $\sqrt{T}$ time). 

Regarding the contribution of Y, we also have the following claim,
\begin{claim}
\label{clm::copy_bound-on-Y-one-strip}
With probability at least $1 - \frac{1}{n^{2c+1}}$, for all $\tau'\ge \tau$, for every strip $s$,
$\big|\sum_{d_i; d_j\in s} Y(d_i,\tau,d_j,\tau') \big| \le 2\sqrt{\frac{3n}{2} \ln \left(2 T n^{3c + 1}\right)}$.
\end{claim}

Thus,
\begin{equation}
\label{eqn::init_strip_variance}
    \begin{aligned}
    |\Imb(s,\tau')|\leq \Big[2B+2\sqrt{\frac{3n}{2} \ln \left(2 T n^{3c + 1}\right)}\Big]\sqrt{T}
    \end{aligned}
\end{equation}

We now calculate $B$.
\begin{claim}
\label{clm:: B_for_init}
For any diagonal $d_j$ and any $d_i$ and $d_k$ that lie in the same strip, at any time $\tau \leq \sqrt{T}$,
$$\Pr\bigg[\Big|\sum_{i\leq j\leq k} X(d_j,\tau)\Big| \geq 2\sqrt{\frac{n(1+\epsilon)^2\sqrt{3}\ln (16T\sqrt{T}(\sqrt{T}+1)n^{c+1})}{0.48\sqrt{T}}}\bigg] \leq \frac{1}{4n^{c+1}}.$$
\end{claim}
\begin{proof}
The agents enter with one of $T$ values chosen uniformly at random and are equally likely to be men or women. Hence, for all $\tau\leq \sqrt{T}$ time steps, for each value $v$,
$$\chance\Big[\text{At most $\frac{n(1+\epsilon)}{2T}$ men enter with value $v$}\Big]\geq 1-\tau Te^{-\frac{\epsilon^2n}{6T}}\geq 1-T\sqrt{T} e^{-\frac{\epsilon^2n}{6T}}.$$
Call this event $\mathcal E_m$.
Similarly,
$$\chance\Big[\text{At most $\frac{n(1+\epsilon)}{2T}$ men enter with value $v$}\Big]\geq 1-\tau Te^{-\frac{\epsilon^2n}{6T}}\geq 1-T\sqrt{T} e^{-\frac{\epsilon^2n}{6T}}.$$
Call this event $\mathcal E_w$.
Henceforth we condition on $\mathcal E_m$ and $\mathcal E_w$.

Consider some time $\tau$. At this time, let $M=\max\{\text{total number of men},\text{total number of women}\}$ while $m=\text{number of men in strip $s$}$ and $w=\text{number of men in strip $s$}$.
Using Lemmas~\ref{lem::negative_dependence_two_sex} and~\ref{lem::match rate}, we obtain the following bound on the deviation from the expected number of the number of men in $s$ matched in a given time step, $\tau$:
\begin{equation}
\label{eqn::init_strip_deviation_bound}
\begin{aligned}
    &\Pr\bigg[\Big|
\begin{array}{l}
    \text{number of men matched}\\
    \hspace*{0.2in}-\Expect[\text{number of men matched}]
\end{array}
    \Big|> \frac{mw\delta}{M}\bigg] \leq 2e^{-{mw\delta^2}/{3M}}.
\end{aligned}
\end{equation}

We will later prove the following claim,

\begin{claim}
\label{claim::early_total_size_lower_bound}
For all time $0\leq t\leq \sqrt{T}$, $0.12nt\leq M\leq nt$ for all $t\leq \sqrt{T}$, with probability at least $1-\frac{1}{5n^{2c+1}}$.
\end{claim}

Call this event $\mathcal E_M$. Henceforth, we further condition on $\mathcal E_M$.

Thus, from equation \ref{eqn::init_strip_deviation_bound}, we obtain
\begin{equation*}
\begin{aligned}
    &\Pr\bigg[\Big|
\begin{array}{l}
    \text{number of men matched}\\
    \hspace*{0.2in}-\Expect[\text{number of men matched}]
\end{array}
    \Big|> \frac{mw\delta}{0.12nt}\bigg] \leq 2e^{-{mw\delta^2}/{3nt}}.
\end{aligned}
\end{equation*}

Setting $\delta=\Big[\frac{3nt}{mw}\ln (n^cA(n,T)) \Big]^{1/2}$, we obtain
\begin{equation*}
\begin{aligned}
    &\Pr\bigg[\Big|
\begin{array}{l}
    \text{number of men matched}\\
    \hspace*{0.2in}-\Expect[\text{number of men matched}]
\end{array}
    \Big|> \sqrt{\frac{mw\sqrt{3}\ln (n^cA(n,T))}{0.12nt}}\bigg] \leq \frac{2}{A(n,T)n^c}.
\end{aligned}
\end{equation*}
We will specify $A(n,T)$ later.
It is easy to see that because of $\mathcal E_m$ and $\mathcal E_w$, $m\leq nt(1+\epsilon)/2\sqrt{T}$ and $w\leq nt(1+\epsilon)/2\sqrt{T}$. So we obtain,
\begin{equation*}
\begin{aligned}
    &\Pr\bigg[\Big|
\begin{array}{l}
    \text{number of men matched}\\
    \hspace*{0.2in}-\Expect[\text{number of men matched}]
\end{array}
    \Big|> \sqrt{\frac{nt(1+\epsilon)^2\sqrt{3}\ln (n^cA(n,T))}{0.48T}}\bigg] \leq \frac{2}{A(n,T)n^c}.
\end{aligned}
\end{equation*}

Since $t\leq \sqrt{T}$,
 
\begin{equation*}
\begin{aligned}
    &\Pr\bigg[\Big|
\begin{array}{l}
    \text{number of men matched}\\
    \hspace*{0.2in}-\Expect[\text{number of men matched}]
\end{array}
    \Big|> \sqrt{\frac{n(1+\epsilon)^2\sqrt{3}\ln (n^cA(n,T))}{0.48\sqrt{T}}}\bigg] \leq \frac{2}{A(n,T)n^c}.
\end{aligned}
\end{equation*}

We can perform the same argument for the women to obtain,

\begin{equation*}
\begin{aligned}
    &\Pr\bigg[\Big|
\begin{array}{l}
    \text{number of women matched}\\
    \hspace*{0.2in}-\Expect[\text{number of women matched}]
\end{array}
    \Big|> \sqrt{\frac{n(1+\epsilon)^2\sqrt{3}\ln (n^cA(n,T))}{0.48\sqrt{T}}}\bigg] \leq \frac{2}{A(n,T)n^c}.
\end{aligned}
\end{equation*}

From this it immediately follows that

\begin{equation*}
\begin{aligned}
    &\Pr\bigg[\Big|\sum_{d\in S}X(d,\tau)
    \Big|> 2\sqrt{\frac{n(1+\epsilon)^2\sqrt{3}\ln (n^cA(n,T))}{0.48\sqrt{T}}}\bigg] \leq \frac{4}{A(n,T)n^c}.
\end{aligned}
\end{equation*}

where $S$ is any subset of the diagonals in strip $s$. We have to consider $\sqrt{T}$ many possible times, $T$ many diagonals $d_j$ and up to $(\sqrt{T}+1)$ many strips. Thus setting $A(n,T)=16T\sqrt{T}(\sqrt{T}+1)n$, proves the claim.
\end{proof}

From equation \eqref{eqn::init_strip_variance} we obtain, for every strip $s$ and $\tau\leq \sqrt{T}$, the following bound on $|\Imb(s,\tau')|$:

\begin{equation*}
    \begin{aligned}
    |\Imb(s,\tau')|&\leq \Big[2B+2\sqrt{\frac{3n}{2} \ln \left(2 T n^{3c + 1}\right)}\Big]\sqrt{T}\\
    &\leq 8\sqrt{n(1+\epsilon)^2\sqrt{T}\ln(16T\sqrt{T}(\sqrt{T}+1)n^{c+1})} + \sqrt{6nT \ln \left(2 T n^{3c + 1}\right)}.
\end{aligned}
\end{equation*}

We are conditioning on $\mathcal E_m$, $\mathcal E_w$ and $\mathcal E_M$. Set $\epsilon= \Big[\frac{6T}{n}\ln(4T\sqrt{T}n^{3c+1})\Big]^{1/2}$. 
We choose constraints so that $\epsilon\leq 1$. So, for every strip $s$ and $\tau\leq \sqrt{T}$,

\begin{equation*}
\begin{aligned}
    &\Pr\bigg[\big|\Imb(s,\tau')\big|\leq 16\sqrt{n\sqrt{T}\ln (16T\sqrt{T}(\sqrt{T}+1)n^{c+1})} + \sqrt{6nT \ln \left(2 T n^{3c + 1}\right)}\bigg]\\
    &\geq \bigg(1-\frac{1}{4n^{c+1}}-\frac{1}{n^{2c+1}}\bigg)\bigg(1-\frac{1}{5n^{2c+1}}-\frac{1}{2n^{2c+1}}\bigg).
\end{aligned}
\end{equation*}

Then we obtain,

\begin{equation*}
\begin{aligned}
    &\Pr\bigg[\big|\Imb(s,\tau')\big|\leq 16\sqrt{n\sqrt{T}\ln (16n^{c+1}T\sqrt{T}(\sqrt{T}+1))} + \sqrt{6nT \ln \left(2 T n^{3c + 1}\right)}\bigg]\\
    &\geq \bigg(1-\frac{1}{4n^{c+1}}-\frac{1}{n^{2c+1}}\bigg)\bigg(1-\frac{7}{10n^{2c+1}}\bigg)]\geq 1-\frac{1}{n^{c+1}}.
\end{aligned}
\end{equation*}

We desire that
$$16\sqrt{n\sqrt{T}\ln (16n^{c+1}T\sqrt{T}(\sqrt{T}+1))} + \sqrt{6nT \ln \left(2 T n^{3c + 1}\right)}\leq \frac{n}{25\sqrt{T}}$$

for which it suffices (by constraints \eqref{eqn::constraints}) that
$$32\cdot 25\cdot T\sqrt{(3c+3)}\ln n\leq \sqrt{n}.$$

or,
$$n\geq (3c+3)(25\cdot 32\cdot T\ln n)^2,$$

which is also provided by the constraints \eqref{eqn::constraints}.
\end{proof}

It now remains to prove Claim \ref{claim::early_total_size_lower_bound}. We proceed exactly as in the proof of Theorem \ref{thm::lb-size}.
\begin{proof} (of Claim~\ref{claim::early_total_size_lower_bound}~). 

Let's consider those agents that enter at times in the range $[0,t]$ for some $t \le \sqrt{T}$. We want to lower bound the number of these agents who are present in the pool for the match at time $t$.

Henceforth, we will only consider men with values in the range $[T+\sqrt{T},2T)$. Among these men, consider those who have been in the pool for $t'$ time, where $0\le t' < \sqrt{T}$.
Let $p_i^{t'}$ be the probability that during their  $t'$-th time step, the men in strip $i$ are offered a match in their own strip. Even if all these men were still present in the matching pool,
\begin{align*}
\chance\Big[\text{\# of these men matched in strip $i$ at age $t'$ } \le \frac{n(1+\delta)(1+\epsilon)\cdot p_i\cdot w_i}{2T} \Big]\geq 1-e^{-\frac{\delta^2np_iw_i}{6T}},
\end{align*}
where $w_i$ is the horizontal width of strip $i$ occupied by these men when aged $t'$.
For every Type 1 strip, $w_i \le \sqrt{T}$.
For the one Type 2 strip, since all values are
at least $T+\sqrt{T}$, 
for ages up to $\sqrt{T}$, $w_i \le \sqrt{T}$.  
By applying $\overline{\mu} = \frac{n(1 + \epsilon) \max\{p_i, \frac{1}{T}\} \sqrt{T}}{2 T}$ in Lemma~\ref{lem::chernoff_bound}, it follows that:

\begin{align*}
\chance\Big[\text{\# of these men matched in strip $i$ at age $t'$ } \le \frac{n(1+\delta)(1+\epsilon)\cdot \max\{p_i, \frac{1}{T}\}}{2\sqrt{T}} \Big]&\geq 1-e^{-\frac{\delta^2n \max\{p_i, \frac{1}{T}\} \sqrt{T}}{6T}} \\
&\geq 1-e^{-\frac{\delta^2n}{6T^{1.5}}},
\end{align*}

The sum of the match probabilities---the $p_i$'s--- is at most $1$. Notice that at any fixed time we only need to consider $\sqrt{T}$ strips, because at any time step, the men we are considering will occupy only $\sqrt{T}$ many strips. This implies $\sum \max\{p_i, \frac{1}{T}\} \leq 1 + \frac{1}{\sqrt{T}}$. Therefore,

$$\chance\Big[\begin{array}{l}\text{\# of these men being matched} \\ \text{over all the strips at age $t'$}\end{array} \le \frac{ (1+\frac{1}{\sqrt{T}})n(1+\delta)(1+\epsilon)}{2\sqrt{T}}\Big]\geq 1-\sqrt{T}\cdot e^{-\frac{\delta^2n}{6T^{1.5}}}.$$

Hence, we can bound the probability of the number of men who entered at time $t -\Delta+1$ and left by time $t$ , for any $\Delta \le t$, as follows:
$$\chance\Big[\begin{array}{l}\text{\# men being matched} \\ \text{in their first $\Delta$ steps}\end{array} \le \frac{ (1+\frac{1}{\sqrt{T}})n\Delta(1+\delta)(1+\epsilon)}{2\sqrt{T}}\Big]\geq 1- \Delta \sqrt{T}\cdot e^{-\frac{\delta^2n}{6T^{1.5}}}.$$

Consequently, we can bound the probability for the number of men that enter in the time interval $[0,t-1]$ and are matched no later than time $t-1$ as follows:

\begin{equation*}
    \begin{aligned}
    &\chance\bigg[\begin{array}{l}\text{\# men who entered and were} \\  \text{matched in a $t-1$ time window}\end{array} \le \frac{(1+\frac{1}{\sqrt{T}})n(1+\delta)(1+\epsilon)(t-1)}{4} \bigg] \geq 1- \frac 12 t^2\sqrt{T} e^{-\frac{\delta^2n}{6T^{1.5}}}.
    \end{aligned}
\end{equation*}

We set $\delta=\Big[\frac{6T^{1.5}}{n} \ln \big(10n^{2c+1} T^{1.5}\big)\Big]^{1/2}$. Note that $t\leq \sqrt{T}$. We already chose $\epsilon= \Big[\frac{6T}{n}\ln(4T\sqrt{T}n^{2c+1})\Big]^{1/2}.$
By constraint~\eqref{eqn::constraints}, $c\geq 1$, $400 \leq T\leq n$, and $n\geq 96T^2(2c+3)\ln n$, 
$\delta\leq 1/4$ and $\epsilon\leq  {1}/{64}$.  This yields the bound:
\begin{align*}
\chance\Big[\begin{array}{l}\text{\# men who entered in a $t-1$} \\ \text{ window being matched }\end{array} \ge \frac{\frac{65}{64} \cdot \frac 54 nt} {4} \Big]
\geq 1-\frac{1}{20n^{2c+1}}.
\end{align*}

Since we have been conditioning on $\mathcal E$, this bound holds with probability at least $1-\frac{1}{10n^{2c+1}}$.
The same bound applies to the women.

Recalling that we excluded the men with values less than $T+\sqrt{T}$,
this yields the following lower bound on the total population size, at time t:
$$nt - \frac{n(1 + \epsilon)}{\sqrt{T}} - 0.635nt\geq 0.25 nt,$$
with probability at least $1-\frac{1}{5n^{2c+1}}.$

Thus $0.12 nt\leq M\leq nt$ for all $t\leq T$, with probability at least $1-\frac{1}{5n^{2c+1}}$, which proves the claim.
\end{proof}

%% file: appendix_data.tex
\section{Further Data from Numerical Simulations}
\label{appn::data}
Here we provide the average population size and average loss we obtained from every run of our numerical simulations.\footnote{every run was for 2000 iterations}.
\subsection{Reasonable Strategy}
\begin{itemize}
    \item Discrete Setting:\\
    $n=500$, $T=100:$
    \begin{center}
 \begin{tabular}{||c c c||} 
 \hline
 Run number & Average population size & Average Loss/T\\ [0.5ex] 
 \hline\hline
 1 & 1544.1 & 10.07\\ 
 \hline
 2 & 1507.7 & 9.96\\
 \hline
 3 & 1558.7 & 10.13\\
 \hline
 4 & 1458.1 & 9.8\\
 \hline
 5 & 1456.4 & 9.82\\ 
 \hline
 6 & 1565.5 & 10.16\\
 \hline
 7 & 1487.6 & 9.91\\
 \hline
 8 & 1567 & 10.17\\
 \hline
 9 & 1499.5& 9.93\\
 \hline
 10 & 1508.3 & 9.99\\
 \hline
\end{tabular}
\end{center}
\noindent
Average Population Size $= 1511.7\pm 3.7\%$.
\newline
Average Loss/T $= 9.99\pm 1.9\%$.
\item Continuum Setting ($=500$, $T=100$):\\
\noindent
Average Population Size $=1180.9$.
\newline
Average Loss/T $= 8.89$.

\end{itemize}

\subsection{Modified Reasonable Strategy}
\begin{itemize}
    \item Discrete Setting:
    \begin{itemize}
        \item $n=500$, $T=100$:\\
 \begin{center}
 \begin{tabular}{||c c c||} 
 \hline
         Run number & Average population size & Average Loss/T\\ [0.5ex] 
 \hline\hline
 1 & 4733.1 &   22.04\\ 
 \hline
 2 &  4735 & 22.06\\
 \hline
 3 &  4666.8 & 21.82\\
 \hline
 4 &  4827.4 & 22.38\\
 \hline
 5 &  4726 & 22.04\\ 
 \hline
 6 &  4685.4 & 21.87\\
 \hline
 7 &  4785.1 & 22.25\\
 \hline
 8 &  4732.8 & 22.06\\
 \hline
 9 & 4681.4 & 21.89\\
 \hline
 10 & 4721.2 & 22.02\\
 \hline
\end{tabular}
\end{center}
 \smallskip
\noindent
Average Population Size $= 4747.1\pm 1.7\%$.
\newline
Average Loss/T $= 22.1\pm 1.7\%$.

        \item $n=500$, $T=200$:\\
 \begin{center}
 \begin{tabular}{||c c c||} 
 \hline
                 Run number & Average population size & Average Loss/T\\ [0.5ex] 
 \hline\hline
 1 &  6972.6& 32.39\\ 
 \hline
 2 &  7194.9& 33.17\\
 \hline
 3 &  7210.3& 33.19\\
 \hline
 4 &  7148.8& 33.01\\
 \hline
 5 &  7336.1& 33.67\\ 
 \hline
 6 &  7085.7& 32.79\\
 \hline
 7 &  7153.2& 32.97\\
 \hline
 8 &  7119.2& 32.94\\
 \hline
 9 & 7195.2& 33.11\\
 \hline
 10 & 7086.6 & 32.79\\
 \hline
\end{tabular}
\end{center}
\smallskip
\noindent
Average Population Size $= 7154.4\pm 2.6\%$.
\newline
Average Loss/T $= 33.03\pm 2\%$.
        
        \item $n=500$, $T=300$:\\
        
 \begin{center}
 \begin{tabular}{||c c c||} 
 \hline
                 Run number & Average population size & Average Loss/T\\ [0.5ex] 
 \hline\hline
 1 &  8963.4& 41.25\\ 
 \hline
 2 &  8930.5& 41.21\\
 \hline
 3 &  8937& 41.17\\
 \hline
 4 &  8744.5& 40.49\\
 \hline
 5 &  8889.2& 41.07\\ 
 \hline
 6 &  8798.6& 40.75\\
 \hline
 7 &  8854.8& 40.89\\
 \hline
 8 &  8783.7& 40.68\\
 \hline
 9 & 8682.9& 40.34\\
 \hline
 10 & 8705.3& 40.37\\
 \hline
\end{tabular}
\end{center}
\smallskip
\noindent
Average Population Size $= 8823.2\pm 1.6\%$.
\newline
Average Loss/T $= 40.8\pm 1.2\%$.
        
        \item $n=500$, $T=400$:\\
        
 \begin{center}
 \begin{tabular}{||c c c||} 
 \hline
                 Run number & Average population size & Average Loss/T\\ [0.5ex] 
 \hline\hline
 1 &  10554.7& 48.4\\ 
 \hline
 2 &  10281& 47.22\\
 \hline
 3 &  10060.9& 46.52\\
 \hline
 4 &  10168.1& 47\\
 \hline
 5 &  10748.9& 48.9\\ 
 \hline
 6 &  10295.3& 47.39\\
 \hline
 7 &  10185.9& 46.86\\
 \hline
 8 &  10094.6& 46.75\\
 \hline
 9 & 10555.1& 48.3\\
 \hline
 10 & 10185.3& 46.99\\
 \hline
\end{tabular}
\end{center}
\noindent
\smallskip
Average Population Size $= 10404.9 \pm 3.4\%$.
\newline
Average Loss/T $= 47.71\pm 2.5\%$.
        
        \item $n=500$, $T=500$:\\
        
 \begin{center}
 \begin{tabular}{||c c c||} 
 \hline
                 Run number & Average population size & Average Loss/T\\ [0.5ex] 
 \hline\hline
 1 &  11186.2& 52.09\\ 
 \hline
 2 &  11268.1& 52.38\\
 \hline
 3 &  11530& 53.44\\
 \hline
 4 &  12209.8& 55.8\\
 \hline
 5 &  12198.1& 55.63\\ 
 \hline
 6 &  12027.3& 55.07\\
 \hline
 7 &  11710.77& 54.03\\
 \hline
 8 &  11372.4& 52.78\\
 \hline
 9 & 11493.2& 53.27\\
 \hline
 10 & 11378.5& 52.77\\
 \hline
\end{tabular}
\end{center}
\smallskip
\noindent
Average Population Size $= 11698\pm 4.4\%$.
\newline
Average Loss/T $= 53.95\pm 3.5\%$.
        
    \end{itemize}
    In summary:\\
 \begin{center}
 \begin{tabular}{||c c c c||} 
 \hline
 n & T & Average population size & Average Loss/T\\ [0.5ex] 
 \hline\hline
 500 & 100& $4747.1\pm 1.7\%$ & $22.1\pm 1.7\%$\\ 
 \hline
 500 & 200& $7154.4\pm 2.6\%$ & $33.03\pm 2\%$\\ 
  \hline
 500 & 300& $8823.2\pm 1.6\%$ & $40.8\pm 1.2\%$\\ 
  \hline
 500 & 400&  $10404.9\pm 3.4\%$.&  $47.71\pm 2.5\%$\\
  \hline
 500 & 500&  $11698\pm 4.4\%$&  $53.95\pm 3.5\%$\\
  \hline
\end{tabular}
\end{center}

    \item Continuum Setting ($n=500$, $T=100$):\\
    \noindent
Average Population Size $= 4484.8$
\newline
Average Loss/T $= 20.85$.
    
\end{itemize}